\documentclass[twoside]{article}
\usepackage{amsfonts, amsthm, amssymb, mathdots}
\usepackage[mathcal]{eucal}
\usepackage[dvips]{graphicx}
\usepackage{subfigure}

\newtheorem{theorem}{Theorem}[section]
\newtheorem{lemma}[theorem]{Lemma}
\newtheorem{prop}[theorem]{Proposition}
\newtheorem{cor}[theorem]{Corollary}

\def\l{\lambda}
\def\a{\alpha}
\def\b{\beta}

\def\g{\gamma}
\def\d{\delta}
\def\e{\varepsilon}
\def\r{\rho}
\def\s{\sigma}
\def\t{\tau}
\def\z{\zeta}

\def\R{\mathbb{R}}
\def\S{\mathbb{S}}
\def\C{\mathbb{C}}
\def\N{\mathbb{N}}
\def\Z{\mathbb{Z}}

\def\M{\mathcal{M}}

\def\Hab{H_{\alpha, \beta}}
\def\H0b{H_{0, \beta}}
\def\Ha0{H_{\alpha, 0}}

\def\mbf2{\mathbf{2}}

\def\sk{\sin\frac{k\pi}{N}}

\def\sl{\sin\frac{l\pi}{N}}
\def\sm{\sin\frac{m\pi}{N}}

\def\1N1{1 \leq k \leq N-1}

\def\zdir{\mathcal{Z}_{Dir}}
\def\thmz{\Theta_{\mathcal{Z}}}
\def\sz{S_{\mathcal{Z}}}
\def\Habd{H_{\a,\b}^{D}}
\def\H0bd{H_{0,\b}^{D}}
\def\Qabd{Q^{D}_{\a,\b}}
\def\Qbd{Q^{D}_{0,\b}}
\def\Qad{Q^{D}_{\a,0}}
\def\Qd{Q^{D}_{0,1}}
\def\Qtoda{Q^{D}_{\a,\a^2}}

\def\fn4{\frac{N}{4}}

\begin{document}

\title{Resonant normal form for even periodic FPU chains}
\author{Andreas Henrici\footnote{Supported in part by the Swiss National Science Foundation} \and Thomas Kappeler\footnote{Supported in part by the Swiss National Science Foundation, the programme SPECT and the European Community through the FP6 Marie Curie RTN ENIGMA (MRTN-CT-2004-5652)}}


\maketitle

\begin{abstract}
In this paper we investigate periodic FPU chains with an even number of particles. We show that near the equilibrium point, any such chain admits a \emph{resonant} Birkhoff normal form of order four which is \emph{completely integrable} - an important fact which helps explain the numerical experiments of Fermi, Pasta, and Ulam. We analyze the moment map of the integrable approximation of an even FPU chain. Unlike in the case of odd FPU chains these integrable systems (generically) exhibit hyperbolic dynamics. As an application we prove that any FPU chain with Dirichlet boundary conditions admits a Birkhoff normal form up to order four and show that a KAM theorem applies.\footnote{2000 Mathematics Subject Classification: 37J10, 37J40, 70H08}
\end{abstract}


\section{Introduction} \label{introduction}

In this paper we consider FPU chains with an even number $N$ of particles of equal mass, normalized to be one. Such chains have been introduced by Fermi, Pasta, and Ulam \cite{fpu}, as models to test numerically the principle of thermalization as $N$ grows larger and larger. A FPU chain consists of a string of particles moving on the line or the circle interacting only with their nearest neighbors through nonlinear springs. Its Hamiltonian is given by
\begin{equation} \label{hvgendef}
  H_V = \frac{1}{2} \sum_{n=1}^N p_n^2 + \sum_{n=1}^N V(q_n - q_{n+1}),
\end{equation}
where $V: \R \to \R$ is a smooth potential. The corresponding Hamiltonian equations read $(1 \leq n \leq N)$
\setlength\arraycolsep{2pt}{\begin{eqnarray*}
  \dot{q}_n & = & \partial_{p_n} H_V = p_n, \\
  \dot{p}_n & = & -\partial_{q_n} H_V = -V'(q_n - q_{n+1}) + V'(q_{n-1} - q_n).
\end{eqnarray*}}
Here $q_n$ denotes the displacement of the $n$'th particle from its equilibrium position and $p_n$ is its momentum. If not stated otherwise, we assume periodic boundary conditions
\begin{equation} \label{perboundcond}
  (q_{i+N}, p_{i+N}) = (q_i, p_i) \quad \forall i \in \{ 0,1 \}.
\end{equation}

Without loss of generality, the potential $V: \R \to \R$ is assumed to have a Taylor expansion at $0$ of the form
\begin{equation} \label{potentialdef}
  V(x) = \kappa \left( \frac{1}{2} x^2 - \frac{\a}{3!} x^3 + \frac{\b}{4!} x^4 + \ldots \right),
\end{equation}
where $\kappa$ is the (linear) spring constant normalized to be $1$ and $\a, \b \in \R$ are parameters measuring the strength of the nonlinear interaction. The minus sign in front of the parameter $\alpha$ in the expansion (\ref{potentialdef}) turns out to be convenient for later computations. Substituting the expression (\ref{potentialdef}) for $V$ into (\ref{hvgendef}), the corresponding expansion of $H_V$ is given by
\begin{equation} \label{hvspecialdef}
  H_V = \frac{1}{2} \sum_{n=1}^N p_n^2 + \frac{1}{2} \sum_{n=1}^N (q_{n\!+\!1} \!-\! q_n)^2 + \frac{\a}{3!} \sum_{n=1}^N (q_{n\!+\!1} \!-\! q_n)^3 + \frac{\b}{4!} \sum_{n=1}^N (q_{n\!+\!1} \!-\! q_n)^4 + \ldots.
\end{equation}
For any FPU chain, the total momentum $P = \frac{1}{N} \sum_{n=1}^N p_n$ is an integral of motion,
and therefore the center of mass $Q = \frac{1}{N} \sum_{n=1}^N q_n$ evolves with constant velocity. Hence any FPU chain can be viewed as a family of Hamiltonian systems of $2N-2$ degrees of freedom, parametrized by the vector of initial conditions $(Q,P) \in \R^2$ with Hamiltonian independent of $Q$. In particular, for $N=2$ any FPU chain is integrable, and hence we will concentrate on the case $N \geq 3$. Further note that for any vector $(Q,P) \in \R^2$, the origin in $\R^{N-2}$ is an equilibrium point of the corresponding system. The momentum of such an equilibrium point is given by the constant vector $(p_1, \ldots, p_N) = P \, (1, \ldots, 1)$.


The frequencies $(\omega_k^0)_{\1N1}$ of the linearization of an arbitrary FPU chain at $(q,p) = (0,0)$ can easily computed to be
        \[ \omega_k^0 = 2 \sin \frac{k\pi}{N}.
\]
The corresponding resonance lattice is given by (see Appendix A of \cite{ahtk4})
\begin{displaymath}
  \left\{ l = (l_1, \ldots, l_{N-1}) \in \Z^{N-1} \Big| \sum_{k=1}^{N-1} l_k \, \sin \frac{k \pi}{N} = 0 \right\}
\end{displaymath}
and generated by the vectors $l^{(k)}$, $\1N1$, defined by $l^{(k)} = e_k - e_{N-k}$, where $e_i$, $1 \leq i \leq N-1$, denotes the standard basis in $\R^{N-1}$. 

For any point $(x,y) \in \R^{2N-2}$ introduce the variables $I = (I_k)_{\1N1}$, where
\begin{equation} \label{actiondef}
I_k = \frac{1}{2}(x_k^2 + y_k^2),
\end{equation}
as well as $M = (M_k)_{\1N1}$, $J = (J_k)_{\1N1}$, and $L = (L_k)_{\1N1}$. They are defined on $\R^{2N-2}$ with values in $\R^{N-1}$ and given by
\begin{equation} \label{jkmkdef}
  M_k \!=\! \frac{1}{2} \! \left( x_k y_{N\!-\!k} \!-\! x_{N\!-\!k} y_k \right); \; J_k \!=\! \frac{1}{2} \! \left( x_k x_{N\!-\!k} \!+\! y_k y_{N\!-\!k} \right); \; L_k \!=\! \frac{1}{2} \! \left( I_k - I_{N-k} \right).
\end{equation}
Note that for any $\1N1$, $(M_k, J_k, L_k) = (-M_{N-k}, J_{N-k}, -L_{N-k})$ and $I_k I_{N-k} = M_k^2 + J_k^2$, or
\begin{equation} \label{hopfintro}
  \left( \frac{I_k + I_{N-k}}{2} \right)^2 = M_k^2 + J_k^2 + L_k^2,
\end{equation}
i.e. $M_k$, $J_k$, $L_k$ are the Hopf variables expressed in $x_k$, $y_k$, $x_{N-k}$, $y_{N-k}$. They describe the image of the Hopf map from the three-dimensional sphere of radius $\frac{1}{2} (I_k + I_{N-k})$ centered at the orgin of $\R^4$. Further define the function $\Hab: \R^{N-1} \to \R$, given by
\begin{equation} \label{bnfintrotheorem}
  \Hab(I) = 2 \! \sum_{k=1}^{N-1} \! \sk \, I_k + \frac{1}{4N} \sum_{k=1}^{N-1} d_k^+ I_k^2 + \frac{\b \!-\! \a^2}{2N} \!\!\!\!\!\!\!\! \sum_{l \neq m \atop 1 \leq l,m \leq N-1} \!\!\!\!\!\!\!\! \sl \sm I_l I_m
\end{equation}
where
\begin{equation} \label{ckabdef}
d_k^+ \equiv d_k^+(\a,\b) := \a^2 + (\b-\a^2) \sin^2 \frac{k \pi}{N}
\end{equation}
and let
\begin{equation} \label{rabdef}
  R_{\a,\b}(J,M) := \frac{\b - \a^2}{4N} \left( R(J,M) + R_\frac{N}{4}(J,M) \right)
\end{equation}
where
\begin{equation} \label{rdef}
R(J,M) = 4 \sum_{1 \leq k < \frac{N}{4}} \sin \frac{2 k \pi}{N} \left( J_k J_{\frac{N}{2} - k} - M_k M_{\frac{N}{2} - k} \right)
\end{equation}
and
\begin{equation} \label{rn4def}
  R_\frac{N}{4}(J,M) = \left\{ \begin{array}{ll}
J_\frac{N}{4}^2 - M_\frac{N}{4}^2 & \quad \textrm{if } \frac{N}{4} \in \N \\
0 & \quad \textrm{otherwise.} \end{array} \right.
\end{equation}
Note that for $\a, \b \in \R$ with $\b = \a^2$ - referred to as Toda case -, $R_{\a,\b}$ vanishes.

In \cite{ahtk4} we proved that even FPU chains admit a resonant normal form near the equilibrium point.

\begin{theorem} \label{bnffputheoremeven}
If $N \! \geq \! 4$ is even, there are canonical coordinates $(x_k, y_k)_{1 \! \leq \! k \leq \! N-1}$ so that the Hamiltonian of \emph{any} FPU chain, when expressed in these coordinates, takes the form $H_V^{trunc}(I,J,M) + O(|(x,y)|^5)$ where
\begin{equation} \label{evenfpuformula}
H_V^{trunc} = \frac{N P^2}{2} + \Hab(I) - R_{\a,\b}(J,M)
\end{equation}
with $\a, \b \in \R$ as in (\ref{potentialdef}) and where $\Hab(I)$ and $R_{\a,\b}(J,M)$ are given by (\ref{bnfintrotheorem}) and (\ref{rabdef}), respectively.
\end{theorem}


Note that in the case $\b \neq \a^2$, the Hamiltonian $H_V$ cannot be transformed into Birkhoff normal form up to order $4$ due to resonances. Nevertheless, the Hamiltonian truncated at order $4$, $H_V^{trunc}$, 
can be proved to be completely integrable. The form of the resonance lattice introduced above suggests that $I_k + I_{N-k}$ ($1 \leq k \leq \frac{N}{2}$) are integrals of $H_V^{trunc}$ in involution. To find the remaining commuting integrals we express $\Hab(I)$ in terms of $I_k + I_{N-k} \; (1 \leq k \leq \frac{N}{2})$ and a remainder term,
\begin{equation} \label{habintrepr}
\Hab(I) = H^{(2)}(I) + H^{(4)}_{\a,\b}(I) + \frac{1}{2N} \sum_{k=1}^{\frac{N}{2}-1} d_k^- I_k I_{N-k}
\end{equation}
where
  \[ H^{(2)}(I) = 2 \sum_{k=1}^{\frac{N}{2}-1} \sk (I_k + I_{N-k}) + 2 I_\frac{N}{2},
\]
\begin{eqnarray*}
  H^{(4)}_{\a,\b}(I) & = & \frac{1}{4N} \sum_{k=1}^{\frac{N}{2}-1} d_k^+ (I_k+I_{N-k})^2 + \frac{\b}{4N} I_\frac{N}{2}^2 \\
&& \quad + \frac{\b - \a^2}{N} I_\frac{N}{2} \sum_{k=1}^{\frac{N}{2}-1} \sk (I_k+I_{N-k}) \\
&& \quad + \frac{\b - \a^2}{2N} \sum_{1 \leq k,l < \frac{N}{2} \atop k \neq l} \sk \sl (I_k + I_{N-k}) (I_l + I_{N-l}),
\end{eqnarray*}
and
  \[ d_k^- \equiv d_k^-(\a,\b) := -\a^2 + (\b-\a^2) \sin^2 \frac{k \pi}{N}.
\]
By (\ref{hopfintro}), one has $I_k I_{N-k} = J_k^2 + M_k^2$ for any $1 \leq k \leq \frac{N}{2}-1$ so that the remainder term $\frac{1}{2N} \sum_{k=1}^{\frac{N}{2}-1} d_k^- I_k I_{N-k}$ in (\ref{habintrepr}) can be written as
  \[ \frac{1}{2N} \left( \sum_{1 \leq k < \frac{N}{4}} \left( d_k^- (J_k^2 + M_k^2) + d_{\tilde{k}}^- (J_{\tilde{k}}^2 + M_{\tilde{k}}^2) \right) + d_{\frac{N}{4}}^- \left( J_{\frac{N}{4}}^2 + M_{\frac{N}{4}}^2 \right) \right),
\]
where the latter term is defined to be $0$ if $\frac{N}{4} \notin \Z$ and $\tilde{k} \equiv \tilde{k}(k) = \frac{N}{2}-k$. Combined with the expression (\ref{rabdef}) for $R_{\a,\b}(J,M)$ the Hamiltonian $H_V^{trunc}$ in (\ref{evenfpuformula}) then takes the form
\begin{equation} \label{h4truncrepr}
  H_V^{trunc} = \frac{N P^2}{2} + H^{(2)}(I) + H^{(4)}_{\a,\b}(I) + \frac{1}{2N} \sum_{1 \leq k \leq \fn4} K_k(I,J,M)
\end{equation}
where for $1 \leq k < \frac{N}{4}$
%
\begin{equation}
  K_k(I,J,M) = d_k^- (J_k^2 \!+\! M_k^2) + d_{\tilde{k}}^- (J_{\tilde{k}}^2 \!+\! M_{\tilde{k}}^2) - 2(\b \!-\! \a^2) \sin \frac{2 k \pi}{N} (J_k J_{\tilde{k}} \!-\! M_k M_{\tilde{k}}) \label{klformula2}
\end{equation}
and
\begin{equation} \label{kn4formula2}
  K_{\frac{N}{4}}(J,M) = \left\{ \begin{array}{ll}
-\a^2 J_\frac{N}{4}^2 + (\b - 2\a^2) M_\frac{N}{4}^2 & \quad \textrm{if } N \equiv 0 \textrm{ mod } 4 \\
0 & \quad \textrm{otherwise.} \end{array} \right.
\end{equation}

\begin{theorem} \label{fpuliouvint}
Let $N \geq 4$ be an even integer. Then the truncated FPU Hamiltonian $H_V^{trunc}$ given by (\ref{evenfpuformula}) is completely integrable. If the expansion of $V$ in (\ref{potentialdef}) satisfies $(\a,\b) \neq (0,0)$, then the following $N-1$ quantities are functionally independent integrals in involution:
  \[ (I_k + I_{N-k})_{1 \leq k \leq \frac{N}{2}}, \; (I_k + I_{\frac{N}{2}+k})_{1 \leq k < \frac{N}{4}}, \; (K_k)_{1 \leq k \leq \frac{N}{4}}.
\]
\end{theorem}

\emph{Remark:} In the case $\b = \a^2$, the integrals $K_k(I,J,M)$ only depend on the action variables $I$, as $M_j^2 + J_j^2 = I_j I_{\frac{N}{2}-j}$ for any $1 \leq j < \frac{N}{2}$.


\vspace{.4cm}

In sections \ref{folmain} and \ref{folapphhgk} we present a detailed analysis of the geometry of the moment map of the integrable system of Theorem \ref{fpuliouvint}. In particular we show that whenever $\b \neq \a^2$, then this integrable system exhibits hyperbolic dynamics.

In section \ref{fpudir}, we will use Theorem \ref{bnffputheoremeven} to show that any FPU chain with Dirichlet boundary conditions admits a Birkhoff normal form up to order $4$ by recognizing such a system as an invariant submanifold of a periodic FPU chain. Consider a chain with $N'$ ($N' \geq 3$, not necessarily even) moving particles and Hamiltonian given by
\begin{displaymath} 
  H_V^{D} = \frac{1}{2} \sum_{n=1}^{N'} p_n^2 + \sum_{n=1}^{N'} V(q_n - q_{n+1}).
\end{displaymath}
Assume that its endpoints are fixed, i.e.
\begin{equation} \label{dirboundcond}
  q_0 = q_{N'+1} = 0.
\end{equation}

\begin{theorem} \label{bnfdirichlettheorem}
Any FPU chain with $N' \geq 3$ moving particles and Dirichlet boundary conditions admits a Birkhoff normal form of order $4$, i.e. there are canonical coordinates $(x_k, y_k)_{1 \leq k \leq N'}$ so that $H_V^{D}$ takes the form
\begin{displaymath}
\frac{(N'+1)P^2}{2} + \Habd(I) + O(|(x,y)|^5)
\end{displaymath}
where $I = (I_1, \ldots, I_{N'})$ is given by (\ref{actiondef}), $\a$, $\b$ are as in (\ref{potentialdef}), and $\Habd(I)$ is of the form
\begin{eqnarray}
  2 \sum_{k=1}^{N'} s_k I_k & + & \frac{1}{16(N'+1)} \sum_{k=1}^{N'} (\a^2 + 3 (\b - \a^2) s_k^2) I_k^2 \quad \underbrace{ + \; \frac{\b - \a^2}{32(N'+1)} I_\frac{N}{4}^2}_{\textrm{only if }\frac{N}{4} \in \N} \nonumber\\
  && + \; \frac{\b - \a^2}{16(N'+1)} \left( \sum_{l \neq m \atop 1 \leq l,m \leq N'} 4 s_l s_m I_l I_m - \sum_{k=1}^{N'} s_{2k} I_k I_{N'+1-k} \right) \! . \; \label{bnfdirichletformula}
\end{eqnarray}
Note that in (\ref{bnfdirichletformula}), the numbers $s_k = \sin \frac{k \pi}{N} = \sin \frac{k \pi}{2N'+2}$ for any $1 \leq k \leq N'$ are pairwise different.
\end{theorem}

\begin{cor}
Near the equilibrium state, any FPU chain with $N'$ moving particles and Dirichlet boundary conditions can be approximated up to order $4$ by an integrable system of $N'$ harmonic oscillators which are \emph{coupled} at fourth order except if $\b = \a^2$.
\end{cor}

Denote by $\Qabd$ the Hessian of $\Habd(I)$ at $I=0$. Note that $\Qabd$ is an $N' \times N'$ matrix which only depends on the parameters $\a$ and $\b$.

\begin{theorem} \label{bnfdirprop}
  \begin{itemize}
  \item[(i)] For any given $\a \in \R \setminus \{ 0 \}$, $\det (\Qabd)$ is a polynomial in $\b$ of degree $N'$ and has $N'$ real zeroes (counted with multiplicities). When listed in increasing order, the zeroes $\b_k = \b_k(\a)$ satisfy
\begin{displaymath}
  \b_1 \leq \ldots \leq \b_{\ulcorner \frac{N'+1}{2} \urcorner} < \a^2 < \b_{\ulcorner \frac{N'+3}{2} \urcorner} \leq \ldots \leq \b_{N'}.
\end{displaymath}
Moreover index$(\Qabd)$, defined as the number of negative eigenvalues of $\Qabd$, is given by
\begin{displaymath}
  \textrm{index} \, (\Qabd) = \left\{  \begin{array}{ll}
\ulcorner \frac{N'+1}{2} \urcorner & \quad \textrm{for } \b < \b_1 \\
0 & \quad \textrm{for } \b_{\ulcorner \frac{N'+1}{2} \urcorner} < \b < \b_{\ulcorner \frac{N'+3}{2} \urcorner} \\
\llcorner \frac{N'-1}{2} \lrcorner & \quad \textrm{for } \b > \b_{N'}
\end{array} \right.
\end{displaymath}
\item[(ii)] For $\a = 0$, $\det(\Qbd)$ is a polynomial in $\b$ of degree $N'$, and $\b = 0$ is the only zero of $\det(\Qbd)$. It has multiplicity $N'$, and the index of $\Qbd$ is given by
\begin{displaymath}
  \textrm{index} \, (\Qbd) = \left\{  \begin{array}{ll}
\ulcorner \frac{N'+1}{2} \urcorner & \quad \textrm{for } \b < 0 \\
\llcorner \frac{N'-1}{2} \lrcorner & \quad \textrm{for } \b > 0
\end{array} \right.
\end{displaymath}
\end{itemize}
\end{theorem}

\emph{Applications:} Theorems \ref{bnfdirichlettheorem} and \ref{bnfdirprop} allow to apply for any given $\a \in \R$ the classical KAM theorem (see e.g. \cite{poeschelkam}) near the equilibrium point to the FPU chain with Hamiltonian $H_V^{D}$ and Dirichlet boundary conditions for a real analytic potential $V(x) = \frac{1}{2} x^2 - \frac{\a}{3!} x^3 + \frac{\b}{4!} x^4 + \ldots$ with $\b \in \R \setminus \{ \b_1(\a), \ldots, \b_{N'}(\a) \}$. Moreover, as for any given $\a \in \R \setminus \{ 0 \}$, $\Qabd$ is positive definite for $\b_{\llcorner \frac{N'-1}{2} \lrcorner}(\a) < \b < \b_{\llcorner \frac{N'+1}{2} \lrcorner}(\a)$, one can apply Nekhoroshev's theorem (see e.g. \cite{poeschelnekh2}) near the equilibrium point to the FPU chain with Hamiltonian $H_V^{D}$ for $V$ with such $\b$'s. These perturbation results confirm long standing conjectures - see e.g. \cite{beiz}.

In subsequent work we plan to compare orbits of the FPU chain with Hamiltonian $H_V$ with the ones of its integrable approximation $H_V^{trunc}$, with the intention to explain the numerical experiments of Fermi, Pasta, and Ulam.

\emph{Related work:} Theorem \ref{fpuliouvint} and Theorem \ref{bnfdirichlettheorem} improve on earlier results of Rink in \cite{rink01} and \cite{rink06}, respectively, where the case $\a=0$ has been treated. Our approach has been shaped by our work on the Toda lattice \cite{ahtk3}. The latter one, introduced by Toda \cite{toda} and extensively studied in the sequel, is a special FPU chain which is integrable (i.e. the original Hamiltonian $H_V$ - not only $H_V^{trunc}$ - is integrable). It turns out that the same canonical transformations which near the equilibrium bring the Toda lattice into Birkhoff normal form can be used for any FPU chain. If $N$ is odd, this transformation brings $H_V$ into Birkhoff normal form up to order $4$; we have studied this case in detail in \cite{ahtk4}.

One of the most important open problems in the field of FPU chains is the investigation of the dynamics of these chains when the number of particles gets larger and larger. It is likely that our results on FPU chains with Dirichlet boundary conditions can be used for this purpose. For recent contributions in this direction see e.g. \cite{bapo1}, \cite{bapo2}.


\emph{Acknowledgement:} It is a great pleasure to thank Bob Rink for valuable comments on earlier versions of this paper, and Gerda Schacher for her help with the graphics.

\section{Proof of Theorem \ref{fpuliouvint}} \label{intproof}


Denote by $\{ \cdot, \cdot \}$ the standard Poisson bracket on $\R^{2N-2}$. In a straightforward way one computes the Poisson brackets between the variables $I, M, J, L \in \R^{N-1}$, given by (\ref{actiondef}) and (\ref{jkmkdef}) (cf. \cite{cuba}, p. $28$):

\begin{lemma} \label{fundbrackets}
The Poisson brackets between the variables $I_k$, $J_k$, $M_k$ ($\1N1$) are given by
\begin{eqnarray}
  \{ I_l, I_k \} & = & \{ J_l, J_k \} = \{ M_l, M_k \} = 0, \label{iijjmm}\\
\{ J_l, I_k \} & = & -M_l (\d_{kl} - \d_{k+l,N}), \label{jlikml}\\
\{ M_l, I_k \} & = & J_l (\d_{kl} - \d_{k+l,N}), \label{mlikjl}\\
\end{eqnarray}
As a consequence, one obtains the following relations between the variables $M_k$, $J_k$, and $L_k$, $\1N1$:
\begin{eqnarray*}
\{ M_k, J_l \} & = & L_l (\d_{k+l,N} - \d_{kl}), \\
\{ J_k, L_l \} & = & M_k (\d_{k+l,N} - \d_{kl}), \\
\{ L_k, M_l \} & = & J_l (\d_{k+l,N} - \d_{kl}).
\end{eqnarray*}
\end{lemma}

First note that the list of functions of Theorem \ref{fpuliouvint},
\begin{equation} \label{1ln4int}
        (I_k + I_{N-k})_{1 \leq k \leq \frac{N}{2}}, \; (I_k + I_{\frac{N}{2}+k})_{1 \leq k < \frac{N}{4}}, \; (K_k)_{1 \leq k \leq \frac{N}{4}},
\end{equation}
contains $N-1$ terms regardless whether $\fn4$ is an integer or not. In addition, for any $1 \leq k < \frac{N}{4}$, the terms $I_k + I_{N-k}, I_{N/2-k} + I_{N/2+k}, I_k + I_{N/2+k}, K_k$ are functions of the eight variables $x_k, y_k, x_{N/2-k}, y_{N/2-k}, x_{N/2+k}, y_{N/2+k}, x_{N-k}$, and $y_{N-k}$, the term $I_{N/2}$ is a function of the two variables $x_{N/2}, y_{N/2}$, and, in the case $N/4 \in \N$, the terms $I_\frac{N}{4} + I_\frac{3N}{4}, K_\frac{N}{4}$ are functions of the four variables $x_{N/4}, y_{N/4}, x_{3N/4}, y_{3N/4}$. Hence we obtain a partition of the $2N-2$ variables $x_1, y_1, \ldots, x_{N-1}, y_{N-1}$ into $\llcorner \frac{N}{4} \lrcorner + 1$ pairwise disjoint sets of variables, and all Poisson brackets between variables of different sets of this partition vanish.

\begin{lemma} \label{involutionlemma}
The $N-1$ functions listed in (\ref{1ln4int}) are pairwise in involution.
\end{lemma}

\begin{proof}
The functions in (\ref{1ln4int}) depend on only one of the $\llcorner \frac{N}{4} \lrcorner + 1$ pairwise disjoint sets of variables. As the Poisson brackets between terms depending on variables of different sets vanish, it remains to check that functions of (\ref{1ln4int}) with the same $k$ are in involution with each other. In view of the definitions (\ref{klformula2}) and (\ref{kn4formula2}) of $K_l$ and $K_{N/4}$, respectively, and taking into account that $(I_k)_{\1N1}$ are pairwise in involution, this amounts to proving that for any $1 \leq l < \frac{N}{4}$,
\begin{eqnarray}
        \{ J_l J_{\frac{N}{2}-l} - M_l M_{\frac{N}{2}-l}, I_l + I_{N-l} \} & = & 0, \label{jjmm1}\\
        \{ J_l J_{\frac{N}{2}-l} - M_l M_{\frac{N}{2}-l}, I_{\frac{N}{2}-l} + I_{\frac{N}{2}+l} \} & = & 0, \label{jjmm2}\\
        \{ J_l J_{\frac{N}{2}-l} - M_l M_{\frac{N}{2}-l}, I_l + I_{\frac{N}{2}+l} \} & = & 0, \label{jjmm3}
\end{eqnarray}
and
\begin{equation} \label{jjmm4}
        \{ J_\frac{N}{4}^2 - M_\frac{N}{4}^2, I_\frac{N}{4} + I_\frac{3N}{4} \} = 0.
\end{equation}
First we note that by (\ref{jlikml}) and (\ref{mlikjl}) one has for any $1 \leq l < \frac{N}{2}$
\begin{equation} \label{jjmmil1}
        \{ J_l J_{\frac{N}{2}-l} - M_l M_{\frac{N}{2}-l}, I_l \} = - J_{\frac{N}{2}-l} M_l - M_{\frac{N}{2}-l} J_l 
\end{equation}
and
\begin{equation} \label{jjmmilnl1}
        \{ J_l J_{\frac{N}{2}-l} - M_l M_{\frac{N}{2}-l}, I_{N-l} \} = J_{\frac{N}{2}-l} M_l + M_{\frac{N}{2}-l} J_l. 
\end{equation}
Since the right hand sides of (\ref{jjmmil1}) and (\ref{jjmmilnl1}) are invariant under exchanging $l$ and $\frac{N}{2}-l$, the same must hold for the left hand sides, and we conclude that
\begin{equation} \label{jjmmil2}
        \{ J_l J_{\frac{N}{2}-l} - M_l M_{\frac{N}{2}-l}, I_{\frac{N}{2}-l} \} = - J_{\frac{N}{2}-l} M_l - M_{\frac{N}{2}-l} J_l
\end{equation}
and
\begin{equation} \label{jjmmilnl2}
        \{ J_l J_{\frac{N}{2}-l} - M_l M_{\frac{N}{2}-l}, I_{\frac{N}{2}+l} \} = J_{\frac{N}{2}-l} M_l + M_{\frac{N}{2}-l} J_l.
\end{equation}
The identities (\ref{jjmm1})-(\ref{jjmm3}) now follow from the appropriate combinations of (\ref{jjmmil1})-(\ref{jjmmilnl2}). In the same fashion, one concludes that (\ref{jjmm4}) holds.
\end{proof}

\begin{proof}[Proof of Theorem \ref{fpuliouvint}]
In view of Lemma \ref{involutionlemma}, it remains to check that the quantities listed in (\ref{1ln4int}) are functionally independent integrals. The independence is easy to verify, and the fact that they are conserved quantities follows from the formula (\ref{h4truncrepr}), showing that $H_V^{trunc}$ can be written as a function of them.
\end{proof}

\section{Foliation by Liouville tori} \label{folmain}

In the next two sections we describe the geometry of the moment map of the truncated resonant normal form (\ref{evenfpuformula}) 
for any even periodic FPU chain $H_V$ with potential $V$ whose expansion (\ref{potentialdef}) satisfies $(\a,\b) \neq (0,0)$. The case $\b = \a^2$ is special as in this case the normal form (\ref{evenfpuformula}) is the Birkhoff normal form of order four of the Toda chain. Its foliation is well known - it is the one of uncoupled harmonic oscillators. Hence we will concentrate on the case $\b \neq \a^2$ only. The special case $\a=0$ has been partially studied by Rink \cite{rink02}. Surprisingly, it turns out that many of his results continue to hold in the general case. Using the notation $\tilde{k} \equiv \tilde{k}(k) = \frac{N}{2}-k$, the integrals of Theorem \ref{fpuliouvint} can be grouped as follows:
\begin{equation} \label{integralslist}
(\mathcal{H}_k, \mathcal{H}_{\tilde{k}}, L_k, K_k)_{1 \leq k < \fn4}, \qquad I_\frac{N}{2}, \qquad \mathcal{H}_\fn4, K_\fn4,
\end{equation}
where for $1 \leq k < \frac{N}{2}$
\begin{displaymath}
  \mathcal{H}_k := I_k + I_{N-k}, \quad L_k := I_k - I_{\tilde{k}}.
\end{displaymath}
(Here we used that $L_k = (I_k + I_{\frac{N}{2}+k}) - (I_{\frac{N}{2}-k} + I_{\frac{N}{2}+k})$ is the difference of two integrals listed in Theorem \ref{fpuliouvint}.)

Using the assumption $\b - \a^2 \neq 0$, we rewrite the integrals $(K_l)_{1 \leq l \leq \fn4}$ as follows. Introduce the bifurcation parameter
\begin{equation} \label{gammadef}
        \g \equiv \g(\a, \b) := \frac{\a^2}{\a^2-\b}
\end{equation}
and note that
  \[ d_k^- = -\a^2 + (\b-\a^2) s_k^2 = (\b-\a^2)(\g + s_k^2).
\]
For any $1 \leq k < \fn4$, one has
        \[ K_k = - d_k^- (J_k^2 \!+\! M_k^2) - d_{\tilde{k}}^- (J_{\tilde{k}}^2 \!+\! M_{\tilde{k}}^2) + 2(\b \!-\! \a^2) s_{2k} (J_k J_{\tilde{k}} \!-\! M_k M_{\tilde{k}}),
\]
whereas for $k = \fn4$,
        \[ K_k = \a^2 J_k^2 - (\b - 2\a^2) M_k^2 = (\a^2 - \b) \left( \g J_k^2 + (1 + \g) M_k^2 \right).
\]
In the sequel, we will for simplicity \emph{omit} the factor $s_{2k} (\a^2 - \b)$ in $K_k$, since it does not influence the geometry of the level sets of the integrals $(K_k)_{1 \leq k \leq \fn4}$.

Each of the $\llcorner \fn4 \lrcorner + 1$ groups of integrals listed in (\ref{integralslist}) depends only on a subset of the variables $\{ (x_k, y_k)_{\1N1} \}$. These subsets form a disjoint partition of $\{ (x_k, y_k)_{\1N1} \}$. More precisely, the following result holds.
\begin{prop}
The phase space $T^*\R^{N-1}$ of the truncated resonant normal form $H_V^{trunc}$ given by (\ref{evenfpuformula}) is the direct sum of invariant symplectic subspaces
        \[ T^*\R^{N-1} = \bigoplus_{0 \leq k \leq \fn4} \mathcal{P}_k
, \]
where
        \[ \mathcal{P}_k = \{ (x_j,y_j)_{1 \leq j \leq N-1} \in T^*\R^{N-1} | x_j = y_j = 0 \; \forall \, j \notin \{ k, N-k, \tilde{k}, N-\tilde{k} \} \}.
\]
The foliation of $T^*\R^{N-1}$ by level sets of the integrals (\ref{integralslist}) is the Cartesian product of the foliations of the $\mathcal{P}_k$.
\end{prop}


We now analyze the foliations of $\mathcal{P}_0$, $\mathcal{P}_\fn4$, and $\mathcal{P}_k$ for $0 < k < \fn4$ separately. Note that $\mathcal{P}_0$ and $\mathcal{P}_k$ ($0 < k < \fn4$) can be canonically identified with $T^*\R$ and $T^*\R^4$, respectively, and $\mathcal{P}_\fn4$ with $T^*\R^2$ (if $\fn4 \in \N$) or $\{ 0 \}$ (if $\fn4 \notin \N$).

\vspace{.4cm}

\noindent
\emph{Foliation of $\mathcal{P}_0$}
One easily sees that $I_\frac{N}{2}$ foliates $T^*\R$ by circles, centered at the origin.

\vspace{.4cm}

\noindent
\emph{Foliation of $\mathcal{P}_{\fn4}$ for $\fn4 \in \Z$} Let us study the geometry of the moment map $\mathcal{M}: T^*\R^2 \to \R^2$ defined by the integrable system with commuting integrals $H \equiv H_{\fn4}$ and $K \equiv K_{\fn4}$. It is convenient to introduce the following notation. Denote the standard coordinates of $T^*\R^2$ by $(x,y) = (x_1,x_2,y_1,y_2)$ and introduce the action variables $I_j = \frac{1}{2} (x_j^2 + y_j^2)$ ($j = 1,2$), as well as the Hopf variables $M,J,L$ given as in (\ref{jkmkdef}),
\begin{displaymath}
  (M,J,L) = \frac{1}{2} (x_1 y_2 - x_2 y_1, x_1 x_2 + y_1 y_2, I_1 - I_2).
\end{displaymath}
Then the moment map $\mathcal{M} = (H,K)$ takes the form
\begin{displaymath}
  H = \frac{1}{2} (I_1 + I_2) \quad \textrm{and} \quad K = (1+\g) M^2 + \g J^2.
\end{displaymath}
As already remarked in (\ref{hopfintro}) one has
\begin{displaymath}
  M^2 + J^2 + L^2 = H^2.
\end{displaymath}
Further, we may replace $K$ by $K_{\g}$ given by
\begin{displaymath}
  K_{\g} := \left\{ \begin{array}{cc} (1+\g) M^2 + \g J^2 & \quad \g \notin \{ 0,1 \} \\ M & \quad \g=0 \\ J & \quad \g=1 \end{array} \right.
\end{displaymath}
First observe that the origin $(x,y) = (0,0)$ of $T^*\R^2$ is the only critical point of $\mathcal{M}$ with rank$(d_{(x,y)} \mathcal{M}) = 0$. Moreover,
\begin{displaymath}
  \mathcal{M}^{-1} \{ (0,0) \} = \{ (0,0) \}.
\end{displaymath}
The critical points $(x,y) \in T^*\R^2 \setminus \{ (0,0) \}$ with rank$(d_{(x,y)} \mathcal{M}) = 1$ are analyzed by symplectic reduction via the Hamiltonian vector field of $H$. On the sphere $\S_{\r}^3 = \{ H = \r^2/4 \}$ of radius $\r > 0$ in $T^*\R^2$ define the Hopf map
\begin{displaymath}
  \mathcal{F}: \S_{\r}^3 \to \S_{r}^2, (x,y) \mapsto (M,J,L) 
\end{displaymath}
where $r = \sqrt{M^2 + J^2 + L^2}|_{\S_{\r}^3} = H|_{\S_{\r}^3} = \frac{\r^2}{4}$. The fibers of $\mathcal{F}$ are circles obtained by the $\S^1$-action of $H$. The reduced system is then given by $(\S_{r}^2, X_{\g})$ where $X_{\g}$ denotes the reduced Hamiltonian vector field induced by $K_{\g}$. To compute $X_{\g}$, note that the equations of motion in the reduced system corresponding to the Hamiltonian $K_{\g}$ are given by
\begin{equation} \label{reducedeqns}
  \frac{d}{dt}  \left( \begin{array}{c} M \\ J \\ L \end{array} \right) = \left( \begin{array}{c} M \\ J \\ L \end{array} \right) \times \left( \begin{array}{c} \partial_M K_{\g} \\ \partial_J K_{\g} \\ \partial_L K_{\g} \end{array} \right).
\end{equation}
Indeed, following the procedure of reduction in section I.5 of \cite{cuba}, formula (\ref{reducedeqns}) follows from
\begin{displaymath}
\dot{w_j} = \{ w_j, K_{\g} \} = \sum_{i=1}^3 \partial_{w_i} K_{\g} \{ w_j,w_i \} \qquad (1 \leq j \leq 3)
\end{displaymath}
and the commutation relations of the variables $(w_1,w_2,w_3) = (M,J,L)$ given by Lemma \ref{fundbrackets}. We then obtain
\begin{displaymath}
  X_{\g} = \left\{ \begin{array}{cc} (-2\g J L, 2(1+\g) M L, -2 M J) & \quad \g \notin \{ -1,0 \} \\ (0,L,-J) & \g=0 \\ (-L,0,M) & \g=-1 \end{array} \right.
\end{displaymath}
It turns out that the foliation of $\S^2_r$ by level sets of $K_{\g}$ depends on the bifurcation parameter $\g$. If $\g=0$, then $(\pm r,0,0)$ are the only two fixed points of $X_0$. They are both elliptic and the level sets of $K_0$ in $\S^2_r \setminus \{ ( \pm r,0,0) \}$ are circles. Similarly, for $\g=-1$, $(0, \pm r,0)$ are the only two fixed points of $X_{-1}$. They are both elliptic and the level sets of $K_{-1}$ in $\S^2_r \setminus \{ (0, \pm r,0) \}$ are circles. Now let us consider the case $\g \in \R \setminus \{ 0,-1 \}$. Then $X_{\g}$ admits six fixed points,
\begin{displaymath}
  (\pm r,0,0), \quad (0, \pm r,0), \quad (0,0, \pm r)
\end{displaymath}
where two of them are hyperbolic and the remaining four elliptic. Note that the corresponding values of $K_{\g}$ are $(1+\g) r^2$, $\g r^2$, and $0$, respectively, and that the two hyperbolic fixed points are contained in the same connected component of the inverse image of $K_{\g}$ in $\S^2_r$. This component consists of two great circles where each of the four half circles is a heteroclinic $X_{\g}$-orbit connecting the two hyperbolic fixed points.
\begin{equation} \label{typeclassific}
  \begin{array}{|c|c|c|}  \hline  & \emph{hyperbolic fixed points} & \emph{critical value} \\ \hline \g < -1 & (\pm r,0,0) & (1+\g) r^2 \\ \hline -1 < \g < 0 & (0,0, \pm r) & 0 \\ \hline \g > 0 & (0, \pm r,0) & \g r^2 \\ \hline \end{array}
\end{equation}
Let us verify the claimed classification of the two fixed points $(0,0, \e r)$ with $\e \in \{ \pm \}$. The other four fixed points are treated in a similar fashion. Near $(0,0, \e r)$ we choose $M,J$ as coordinates of $\S^2_r$. The equations of motion induced by $X_{\g}$ on $\S^2_r$ in these coordinates read
\begin{eqnarray*}
  \dot{M} & = & -\e 2 \g J \sqrt{r^2 - M^2 - J^2}, \\
  \dot{J} & = & \e 2 (1+\g) M \sqrt{r^2 - M^2 - J^2}.
\end{eqnarray*}
If linearized at $(0,0,\e r)$ the corresponding linear system is given by the $2 \times 2$-matrix $\e A$ where
\begin{displaymath}
  A = 2 r \left( \begin{array}{cc} 0 & -\g \\ 1+\g & 0 \end{array} \right).
\end{displaymath}
The eigenvalues of $A$ are given by
\begin{displaymath}
  \l^2 + 4 \g (1 + \g) r^2 = 0 \quad \textrm{or} \quad \l_{1,2} = \pm 2 r \sqrt{-\g (1+\g)}.
\end{displaymath}
As $-\g (1+\g) > 0$ iff $-1 < \g < 0$ it follows that $\l_{1,2}$ are in $\R \setminus \{ 0 \}$ and hence that $(0,0, \pm r)$ are both hyperbolic fixed points for $-1 < \g < 0$ whereas they are both elliptic if $\g < -1$ or $\g > 0$. For $-1 < \g < 0$, the inverse image $K_{\g}^{-1}(\{ 0 \})$ is given by
\begin{eqnarray*}
  K_{\g}^{-1}(\{ 0 \}) & = & \{ (M,J,L) | (1+\g) M^2 + \g J^2 = 0; M^2 + J^2 + L^2 = r^2 \} \\
& = & \{ (M,J,L) | M = \pm \left| \frac{\g}{1+\g} \right|^{1/2} J; M^2 + J^2 + L^2 = r^2 \},
\end{eqnarray*}
whereas for $\g < -1$ or $\g > 0$, $K_{\g}^{-1}(\{ 0 \}) = \{ (0,0, \pm r) \}$.

\vspace{.4cm}

\noindent
\emph{Foliation of $\mathcal{P}_k$ for $0 < k < \fn4$} The foliation of $\mathcal{P}_k$ for $0 < k < \fn4$ is more complicated than in the case $k = \fn4$. We describe it in detail in the subsequent section.

\section{Integrable system $(\mathcal{H}_k, \mathcal{H}_{\tilde{k}}, G_k, K_k)$} \label{folapphhgk}

In this section we present a detailed study of the geometry of the moment map $\mathcal{M}: T^*\R^4 \to \R^4$ defined by the integrable system $\mathcal{P}_k$ with commuting integrals $\mathcal{H}_k$, $\mathcal{H}_{\tilde{k}}$, $G_k$, $K_k$ for $1 \leq k < \fn4$. As in section \ref{folmain} we restrict to FPU chains with potential $V$ whose expansion (\ref{potentialdef}) satisfies $\b \neq \a^2$. We show that in this case the vector field induced by $K_k$ exhibits hyperbolic dynamics. It is convenient to introduce the following notation. Denote the standard coordinates of $T^*\R^4$ by $(x,y)$ with $x = (x_i)_{1 \leq i \leq 4}$ and $y = (y_i)_{1 \leq i \leq 4}$, and introduce the action variables $I_j = \frac{1}{2}(x_j^2 + y_j^2)$ ($1 \leq j \leq 4$), as well as the Hopf variables $(M_i,J_i,L_i)_{1 \leq i \leq 2}$ given by
\begin{eqnarray*}
  (M_1,J_1,L_1) & = & \frac{1}{2} (x_1 y_2 - x_2 y_1, x_1 x_2 + y_1 y_2, I_1 - I_2), \\
  (M_2,J_2,L_2) & = & \frac{1}{2} (x_3 y_4 - x_4 y_3, x_3 x_4 + y_3 y_4, I_3 - I_4).
\end{eqnarray*}
By Lemma \ref{fundbrackets}, the Poisson brackets between the variables $(M_i,J_i,L_i)_{1 \leq i \leq 2}$ are given by
\begin{displaymath}
  \{ M_i, J_i \} = -L_i, \; \{ J_i, L_i \} = -M_i, \; \{ L_i, M_i \} = -J_i
\end{displaymath}
whereas all other brackets vanish.

The moment map $\mathcal{M}$ then takes the form
\begin{displaymath}
  \mathcal{M}: T^*\R^4 \to \R^4, \quad (x,y) \mapsto (H_1,H_2,G,K_{\g})
\end{displaymath}
where
\begin{displaymath}
  H_1 = \frac{1}{2} (I_1 + I_2); \quad H_2 = \frac{1}{2} (I_3 + I_4); \quad G = L_1 - L_2
\end{displaymath}
and where $K_{\g}$ is a scalar multiple of $K_k$, given by
\begin{displaymath}
  K_{\g} = \sum_{i=1}^2 \frac{1}{2} d_{i,\g} (M_i^2 + J_i^2) + (M_1 M_2 - J_1 J_2)
\end{displaymath}
with the coefficients $d_{1,\g}$, $d_{2,\g}$ defined by
\begin{displaymath}
  d_{1,\g} = \frac{\g + s_k^2}{s_{2k}}, \quad d_{2,\g} = \frac{\g + c_k^2}{s_{2k}}
\end{displaymath}
and $s_k = \sin \frac{k \pi}{N}$, $c_k = \cos \frac{k \pi}{N}$. (The definition of the integral $G$ above differs from the one given earlier by a multiple of the integral $H_1 - H_2$ as $I_1 - I_3 = L_1 - L_2 + H_1 - H_2$.)

First note that the origin $(0,0)$ in $T^*\R^4$ is the only critical point of $\M$ with rank$(d\M) = 0$. Moreover, $\M^{-1}\{ (0,0) \} = \{ (0,0) \}$. Next observe that when restricted to $T^*\R^2 \times \{ 0 \}$, one has $G = \frac{1}{2}(I_1 - I_2)$ and $K_{\g} = d_{1,\g} (H_1^2 - L_1^2)$, hence they are functions of $I_1$, $I_2$ alone and $\M|_{T^*\R^2 \times \{ 0 \}}$ may be replaced by the map $(x,y) \mapsto (I_1,I_2,0,0)$. The geometry of the latter map is the one of two uncoupled harmonic oscillators. The subspace $\{ 0 \} \times T^*\R^2$ is treated similarly. It remains to study the restriction of $\M$ to $T^*\R^4 \setminus ((T^*\R^2 \times \{ 0 \}) \cup (\{ 0 \} \times T^*\R^2))$. The Hamiltonian vector fields of $H_1$ and $H_2$ induce a torus action on $T^*\R^2$. The corresponding symplectic reduction is given by the product of two Hopf maps,
\begin{displaymath}
  \mathcal{F}: \S_{\r_1}^3 \times \S_{\r_2}^3 \to \S_{r_1}^2 \times \S_{r_2}^2, (x,y) \mapsto (M_i,J_i,L_i)_{1 \leq i \leq 2}
\end{displaymath}
where for $i = 1,2$, $\S_{\r_i}^3 = \{ H_i = \rho_i^2/4 \}$ is a sphere in $T^*\R^2$ and $r_i = \r_i^2/4 = \sqrt{M_i^2 + J_i^2 +L_i^2}|_{\S_{\r_i}^3} = {H_i}|_{\S_{\r_i}^3}$. 
The fibers of $\mathcal{F}$ are $2$-dimensional tori, obtained by the $\S^1 \times \S^1$-action of $H_1 \times H_2$. The reduced system is then given by $(\S_{r_1}^2 \times \S_{r_2}^2, Y, X_{\g})$, where $Y$ and $X_{\g}$ denote the reduced Hamiltonian vector fields induced by $G$ and $K_{\g}$, respectively. To compute $Y$ and $X_{\g}$, note that the equations of motion in the reduced system, corresponding to a Hamiltonian $H$, are given by
\begin{equation} \label{reducedeqnsk<n4}
  \frac{d}{dt}  \left( \begin{array}{c} M_i \\ J_i \\ L_i \end{array} \right) = \left( \begin{array}{c} M_i \\ J_i \\ L_i \end{array} \right) \times \left( \begin{array}{c} \partial_{M_i} H \\ \partial_{J_i} H \\ \partial_{L_i} H \end{array} \right), \quad i=1,2
\end{equation}
- see section \ref{folmain} for details. We then obtain
\begin{equation} \label{yxgammavf}
  Y = \left( \begin{array}{c} J_1 \\ -M_1 \\ 0 \\ -J_2 \\ M_2 \\ 0 \end{array} \right), \quad X_{\g} = \left( \begin{array}{c} (J_2 - d_{1,\g} J_1) L_1 \\ (d_{1,\g} M_1 + M_2) L_1 \\ -(M_1 J_2 + M_2 J_1) \\ (J_1 - d_{2,\g} J_2) L_2 \\ (d_{2,\g} M_2 + M_1) L_2 \\ -(M_1 J_2 + M_2 J_1) \end{array} \right).
\end{equation}
Further introduce the reduced moment map
\begin{displaymath}
  \M_{\g}: \S_{r_1}^2 \times \S_{r_2}^2 \to \R^2, (M_i,J_i,L_i)_{1 \leq i \leq 2} \mapsto (G, K_{\g}).
\end{displaymath}
We now study the critical points of $\M_{\g}$ with rank$\, (d\M_{\g}) = 0$, i.e. points of $\S_{r_1}^2 \times \S_{r_2}^2$ which are fixed points of both, $Y$ and $X_{\g}$. From the expressions for $Y$ and $X_{\g}$ derived above, one easily sees that there are only four such critical points,
\begin{displaymath}
  (M_i,J_i,L_i)_{1 \leq i \leq 2} = \e (0,0,r_1,0,0, \pm r_2),
\end{displaymath}
where $\e \in \{ \pm \}$. The value of the critical point $\e (0,0,r_1,0,0,-r_2)$ by $\M_{\g}$ is $(\e (r_1 + r_2), 0)$ and
\begin{displaymath}
  \M_{\g}^{-1} \{ (\e (r_1 + r_2), 0) \} = \{ \e (0,0,r_1,0,0,-r_2) \}.
\end{displaymath}
Computing the Jacobian of the reduced vector field $X_{\g}$ at the critical points one sees that they are elliptic fixed points of $X_{\g}$. The values of the other two critical points $\e (0,0,r_1,0,0,r_2)$ by $\M_{\g}$ are $(\e (r_1 - r_2), 0)$. The inverse image of $(\e (r_1 - r_2), 0)$ might have several connected components, depending on the values of $\g$ and the additional bifurcation parameter
\begin{displaymath}
r := \frac{r_1}{r_2} > 0.
\end{displaymath}
Our main results concerning the critical points $\e (0,0,r_1,0,0,r_2)$ are collected in the folloing theorem.
\begin{theorem} \label{kgenfolthm}
Assume that $1 \leq k < \fn4$, $0 < r \leq 1$, $\e \in \{ \pm \}$, and $\g \in \R$. The critical point $\e (0,0,r_1,0,0,r_2)$ of $\M_{\g}$ is a \emph{hyperbolic} fixed point of the vector field $X_{\g}$ if and only if
\begin{equation} \label{hyperbolcondsingle}
        \left| (\g + s_k^2) \sqrt{r} + (\g + c_k^2) \frac{1}{\sqrt{r}} \right| < 2 s_{2k}.
\end{equation}
Otherwise it is an \emph{elliptic} fixed point of $X_{\g}$. If (\ref{hyperbolcondsingle}) is satisfied, the stable and unstable manifolds of $\e (0,0,r_1,0,0,r_2)$ both have dimension two. In the case \mbox{$r<1$}, the connected component of $\M_{\g}^{-1} \{ \e (r_1 - r_2,0) \}$ containing $\e (0,0,r_1,0,0,r_2)$ is a $2$-dimensional torus pinched at $\e (0,0,r_1,0,0,r_2)$ and consists of homoclinic $X_{\g}$-orbits. In the case $r=1$, $\M_{\g}^{-1} \{ (0,0) \}$ is a $2$-dimensional torus pinched at the two points $\pm (0,0,r_1,0,0,r_1)$, and $\M_{\g}^{-1} \{ (0,0) \} \setminus \{ \pm (0,0,r_1,0,0,r_1) \}$ consists of heteroclinic $X_{\g}$-orbits.
\end{theorem}

To prove Theorem \ref{kgenfolthm} we separately treat for any given $1 \leq k < \fn4$ three subsets of the domain of the parameters $\g$ and $r$. The results for these three cases are stated in detail in Propositions \ref{req1prop}, \ref{propgck2eq0}, and \ref{propgck2neq0} below.

\begin{figure}
  \centering
  \subfigure[$k=1$]{
     \label{fig:k=1}
     \includegraphics[width=1.5in]{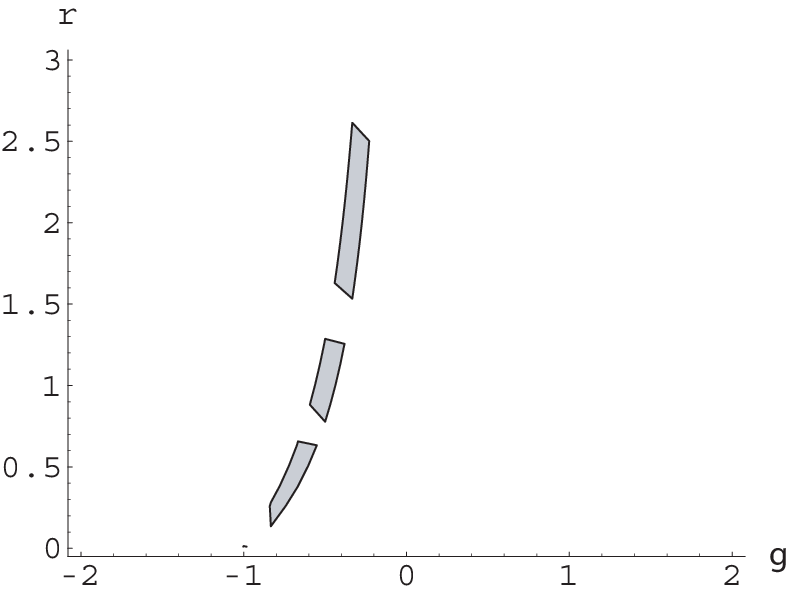}}
  \hspace{.5in}
  \subfigure[$k=5$]{
     \label{fig:k=5}
     \includegraphics[width=1.5in]{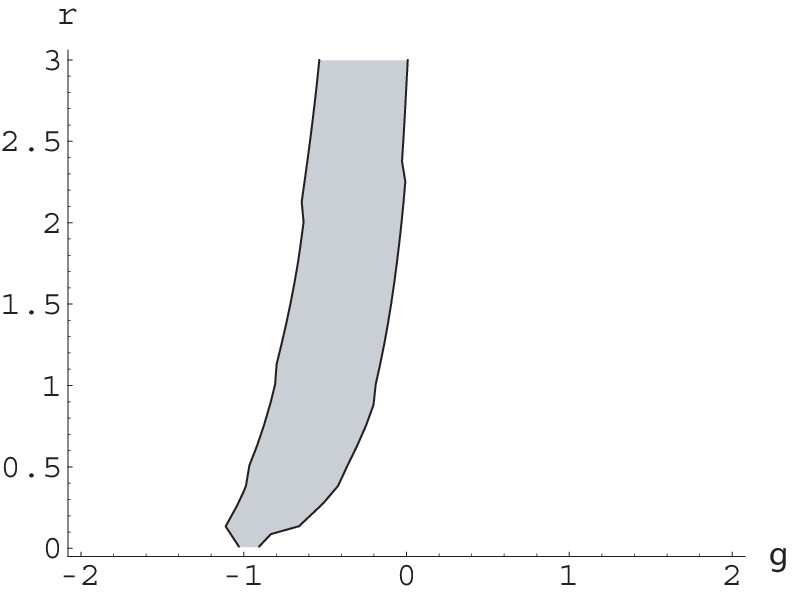}}
\\
  \subfigure[$k=10$]{
     \label{fig:k=10}
     \includegraphics[width=1.5in]{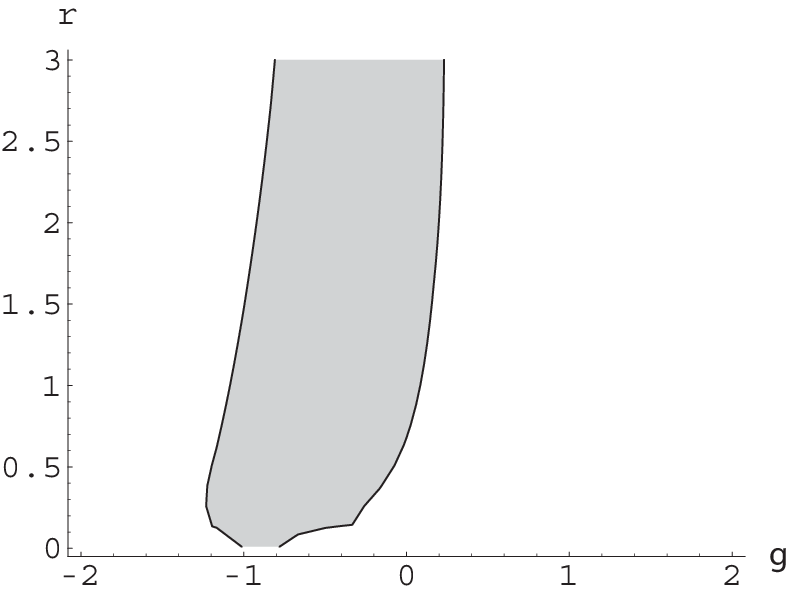}}
  \hspace{.5in}
  \subfigure[$k=24$]{
     \label{fig:k=24}
     \includegraphics[width=1.5in]{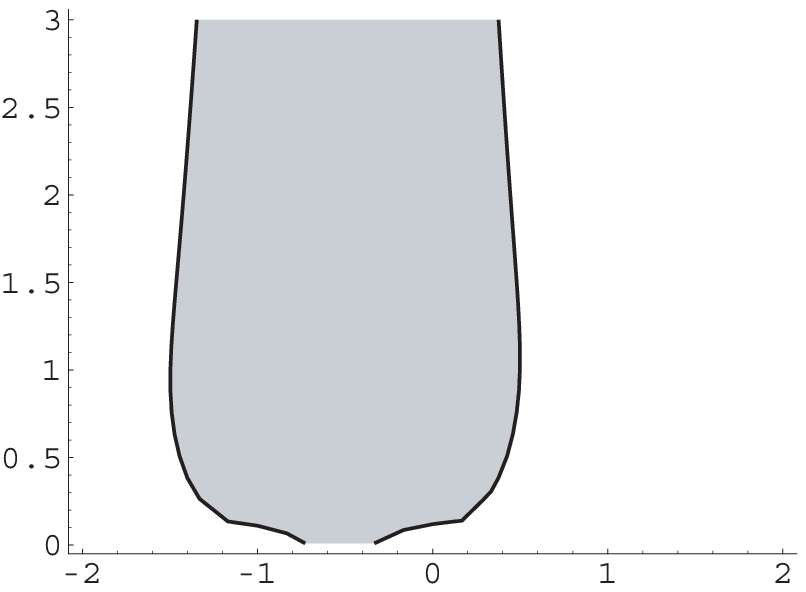}}
  \caption{Subsets of parameters $(\g,r)$ with hyperbolic dynamics of $X_{\g}$ ($N\!=\!100$)}
  \label{fig:discr}
\end{figure}

Note that the inverse image $\M_{\g}^{-1} \{ (\e (r_1 - r_2), 0) \}$ is invariant under the action of the vector field $Y$. The orbits of this action can be easily described,
\begin{equation} \label{orbitsgaction}
  \{ (R_{-\phi}(M_1,J_1), L_1, R_{\phi}(M_2,J_2), L_2) \big| |\phi| \leq \pi \},
\end{equation}
where $R_{\phi}(u,v)$ denotes the image of $(u,v) \in \R^2$ of the rotation $R_{\phi}$ by the angle $\phi$ in counterclockwise orientation. Hence given $L_1$ and $L_2$ with $|L_i| < r_i$ for $i=1,2$ there exists a unique point $(\hat{M}_i, \hat{J}_i, L_i)_{1 \leq i \leq 2}$ on such an orbit with the property that
\begin{displaymath}
  (\hat{M}_1, \hat{J}_1) = \left( \sqrt{r_1^2 - L_1^2}, \, 0 \right) \quad \textrm{and} \quad (\hat{M}_2, \hat{J}_2) = \sqrt{r_2^2 - L_2^2} \; \left( \cos \a, \sin \a \right)
\end{displaymath}
for some $0 \leq \a < 2\pi$. We denote the corresponding $Y$-orbit by $\mathcal{L}(L_1,L_2,\a)$, i.e.
\begin{equation} \label{lcaldef}
  \mathcal{L}(L_1,L_2,\a) = \big\{ \big( R_{-\phi} (\sqrt{r_1^2 \!-\! L_1^2}, 0), L_1, R_{\a+\phi} (\sqrt{r_2^2 \!-\! L_2^2}, 0), L_2 \big) \big| |\phi| \leq \pi \big\}.
\end{equation}
As $G$ and $K_\g$ commute, $K_{\g}$ is invariant along any orbit of the vector field $Y$ and we conclude that
\begin{equation} \label{kgammaalpha}
  K_{\g}((M_i,J_i,L_i)_{1 \leq i \leq 2}) = \frac{1}{2} \sum_{i=1}^2 d_{i,\g} (r_i^2 - L_i^2) + \sqrt{(r_1^2 - L_1^2)(r_2^2 - L_2^2)} \cos \a.
\end{equation}
Let us now determine $\M_{\g}^{-1} \{ (\e (r_1 - r_2), 0) \}$ for $\e = 1$ and $r \leq 1$. (The case where $r > 1$ and/or $\e = -1$ is treated in a similar fashion.) Let $(M_i,J_i,L_i)_{1 \leq i \leq 2}$ be an element of $\M_{\g}^{-1} \{ (r_1 - r_2, 0) \} \setminus \{ (0,0,r_1,0,0,r_2) \}$. Then $r_2 - L_2 = r_1 - L_1$ and $K_{\g}((M_i,J_i,L_i)_{1 \leq i \leq 2}) = 0$. First note that $L_1 < r_1$. Indeed, if $L_1 = r_1$, then $L_2 = r_2$ and $(M_i,J_i) = (0,0)$ for $i=1,2$, contradicting our assumption on the point considered. Hence in the expression for $K_{\g}$ displayed above we can factor out $r_1 - L_1$ and the equation $K_{\g} = 0$ reads
\begin{equation} \label{eqnkg=0}
  0 = d_{1,\g} (r_1 + L_1) + d_{2,\g} (r_2 + L_2) + 2 \sqrt{(r_1 + L_1)(r_2 + L_2)} \cos \a.
\end{equation}
Next let us consider the case where $L_1 = -r_1$. Then $(M_1,J_1,L_1) = -(0,0,r_1)$ and (\ref{eqnkg=0}) reads $(\g + c_k^2) (r_2 + L_2) = 0$. Hence either $L_2 = -r_2$ or $\g = -c_k^2$. In the case $L_2 = -r_2$ it follows from $r_2 - L_2 = r_1 - L_1 = 2 r_1$ that $r_2 = r_1$ and $L_2 = L_1$. As a consequence $(M_2,J_2,L_2) = -(0,0,r_1)$. On the other hand, if $\g = -c_k^2$ and $r_1 < r_2$, then
\begin{displaymath}
-r_2 < r_2 - 2 r_1 = L_2 < r_2 \quad \textrm{and} \quad  M_2^2 + J_2^2 = r_2^2 - L_2^2 = 4 r_1 (r_2 - r_1).
\end{displaymath}
If $L_1 \neq -r_1$ (and hence $L_2 \neq -r_2$ as $r_1 \leq r_2$) we set for $-r_1 < L_1 < r_1$
\begin{equation} \label{qfndef}
  Q(L_1) := \left\{ \begin{array}{cc} \sqrt{\frac{r_1 + L_1}{2 r_2 + L_1 - r_1}} & \quad \textrm{if} \; r_1 < r_2 \\ 1 & \quad \textrm{if} \; r_1 = r_2 \end{array} \right.
\end{equation}
Then $0 < Q < \sqrt{r}$ and for $r<1$, $Q$ is monotonically increasing on $-r_1 < L_1 < r_1$. After division by $\frac{1}{s_{2k}} \sqrt{(r_1 + L_1)(r_2 + L_2)}$ the equation (\ref{eqnkg=0}) reads
\begin{equation} \label{qgeqn}
  (\g + s_k^2) Q + (\g + c_k^2) \frac{1}{Q} + 2 s_{2k} \cos \a = 0.
\end{equation}
To investigate the solutions of (\ref{qgeqn}) we distinguish between three cases: $r=1$, [$\g+c_k^2=0$ and $r<1$], and [$\g+c_k^2 \neq 0$ and $r<1$].

Let us first treat the case $r=1$. Then $Q \equiv 1$ and equation (\ref{qgeqn}) takes the form
\begin{equation} \label{qgeqnreq1}
  2 \g + 1 = -2 s_{2k} \cos \a,
\end{equation}
which is independent of $L_1$.

\begin{prop} \label{req1prop}
Let $\g \in \R$ be arbitrary and assume $1 \leq k < \fn4$ and $r=1$. Then the following statements hold:
\begin{itemize}
\item[(i)] If $|2\g + 1| > 2 s_{2k}$, then
  \begin{displaymath}
    \M_{\g}^{-1} \{ (0,0) \} = \{ \e (0,0,r_1,0,0,r_1) | \e = \pm \},
  \end{displaymath}
and $\pm (0,0,r_1,0,0,r_1)$ are both \emph{elliptic} fixed points of the vector field $X_{\g}$.
\item[(ii)] If $|2\g + 1| < 2 s_{2k}$, then $\M_{\g}^{-1} \{ (0,0) \} = \mathcal{N}_{\a} \cup \mathcal{N}_{-\a}$, where $\a$ is the unique angle satisfying $0 < \a < \pi$ and $2\g + 1 = -2 s_{2k} \cos \a$, and where for any $-\pi \leq \b \leq \pi$
  \begin{displaymath}
    \mathcal{N}_{\b} = \bigcup_{|L_1| \leq r_1 \atop L_2 = L_1} \mathcal{L}(L_1,L_2,\b)
  \end{displaymath}
with $\mathcal{L}(L_1,L_2,\b)$ given by (\ref{lcaldef}). Both points, $\pm (0,0,r_1,0,0,r_1)$, are \emph{hyperbolic} fixed points of $X_{\g}$, and their stable and unstable manifolds have each dimension two. The set $\mathcal{N}_{\a} \setminus \{ \pm (0,0,r_1,0,0,r_1) \}$ consists of heteroclinic $X_{\g}$-orbits from $(0,0,r_1,0,0,r_1)$ to $-(0,0,r_1,0,0,r_1)$, whereas $\mathcal{N}_{-\a} \setminus \{ \pm (0,0,r_1,0,0,r_1) \}$ consists of heteroclinic $X_{\g}$-orbits with opposite direction. Topologically, $\M_{\g}^{-1} \{ (0,0) \}$ is a $2$-dimensional torus, pinched at each of the two fixed points $\pm (0,0,r_1,0,0,r_1)$.
\item[(iii)] If $2\g + 1 = -2 s_{2k}$, then $\a=0$ and $\M_{\g}^{-1} \{ (0,0) \} = \mathcal{N}_{0}$, whereas if $2\g + 1 = 2 s_{2k}$, then $\a = \pi$ and $\M_{\g}^{-1} \{ (0,0) \} = \mathcal{N}_{\pi}$. In both cases, $\pm (0,0,r_1,0,0,r_1)$ are \emph{elliptic} fixed points of $X_{\g}$. On $\mathcal{N}_{0} \cup \mathcal{N}_{\pi}$, any $X_{\g}$-orbit is periodic and coincides with the corresponding $Y$-orbit at least up to orientation.
\end{itemize}
\end{prop}

\begin{proof}
(i) By a straightforward computation one shows that under the given assumptions, both points $\pm (0,0,r_1,0,0,r_1)$ are elliptic fixed points of $X_{\g}$. In view of equation (\ref{qgeqnreq1}), item (i) then easily follows.

(ii) By the discussion preceding Proposition \ref{req1prop} it follows that the inverse image $\M_{\g}^{-1} \{ (0,0) \}$ is given as claimed. Again by a straightforward computation one shows that both fixed points $\pm (0,0,r_1,0,0,r_1)$ are hyperbolic. To see that $\mathcal{N}_{\a} \setminus \{ \pm (0,0,r_1,0,0,r_1) \}$ consists of heteroclinic orbits of the vector field $X_{\g}$, consider the third component $(X_{\g})_3$ of $X_{\g}$ (cf (\ref{yxgammavf})). Any element $(M_i,J_i,L_i)_{1 \leq i \leq 2}$ in $\mathcal{N}_{\a}$ is of the form
\begin{equation} \label{na0repr}
  (M_1,J_1) \!=\! \sqrt{r_1^2 \!-\! L_1^2} (\cos \phi, -\sin \phi), \; (M_2,J_2) \!=\! \sqrt{r_1^2 \!-\! L_1^2} (\cos (\a \!+\! \phi), \sin (\a \!+\! \phi)).
\end{equation}
Thus
\begin{eqnarray}
  (X_{\g})_3 & = & -(M_1 J_2 + M_2 J_1) \nonumber\\
& = & -(r_1^2 - L_1^2) \, (\cos \phi \, \sin (\a + \phi) - \cos (\a + \phi) \sin \phi). \label{req1m1j2m2j1}
\end{eqnarray}
Hence $(X_{\g})_3 = -(r_1^2 - L_1^2) \sin \a < 0$ for any point in $\mathcal{N}_{\a} \setminus \{ \pm (0,0,r_1,0,0,r_1) \}$. As the last component of $X_{\g}$ coincides with the third one, it follows that any $X_{\g}$-orbit on $\mathcal{N}_{\a} \setminus \{ \pm (0,0,r_1,0,0,r_1) \}$ originates from $(0,0,r_1,0,0,r_1)$ and ends in $-(0,0,r_1,0,0,r_1)$. The orbits on $\mathcal{N}_{-\a}$ are analyzed in a similar way.

(iii) Clearly, if $2\g + 1 = -2 s_{2k}$, one has $\M_{\g}^{-1} \{ (0,0) \} = \mathcal{N}_{0}$ and one verifies in a straightforward way that $\pm (0,0,r_1,0,0,r_1)$ are elliptic fixed points of $X_{\g}$. According to (\ref{req1m1j2m2j1}), the third component $(X_{\g})_3$ of $X_{\g}$ vanishes identically on $\mathcal{N}_{0}$. Further, $2\g+ 1 = -2 s_{2k}$ implies that $1 + d_{2,\g} = -1-d_{1,\g}$. In view of (\ref{na0repr}) it then follows that
\begin{displaymath}
  X_{\g} = (1 + d_{2,\g}) L_2 \cdot Y.
\end{displaymath}
The claimed statements for the case $2\g+1 = 2 s_{2k}$ are proved in a similar fashion.
\end{proof}

Next we consider the case where $r<1$ and $\g + c_k^2 = 0$. Then $\g + s_k^2 = -c_{2k}$ and hence
\begin{displaymath}
  d_{1,\g} = -\frac{c_{2k}}{s_{2k}} \quad \textrm{and} \quad d_{2,\g} = 0.
\end{displaymath}
Thus equation (\ref{qgeqn}) takes the form
\begin{equation} \label{qgeqngck2}
  c_{2k} Q(L_1) = 2 s_{2k} \cos \a.
\end{equation}
Note that $0 < c_{2k} < 1$ as $1 \leq k < \fn4$.

\begin{prop} \label{propgck2eq0}
Assume that $1 \leq k < \fn4$, $0<r<1$, and $\g + c_k^2 = 0$. Then the following statements hold:
\begin{itemize}
\item[(i)] If $\sqrt{r} > 2 s_{2k} / c_{2k}$, then the connected component of $\M_{\g}^{-1} \{ (r_1 - r_2,0) \}$ containing the critical point $(0,0,r_1,0,0,r_2)$ consists of this point alone. It is an \emph{elliptic} fixed point of $X_{\g}$.
\item[(ii)] If $\sqrt{r} \leq 2 s_{2k} / c_{2k}$, then
  \begin{displaymath}
    \M_{\g}^{-1} \{ (r_1 - r_2,0) \} = \bigcup_{|L_1| \leq r_1 \atop L_2 = L_1 + r_2 - r_1} \mathcal{L}(L_1,L_2,\a_{L_1}) \cup \mathcal{L}(L_1,L_2,-\a_{L_1})
  \end{displaymath}
where for any $|L_1| \leq r_1$, $\a_{L_1}$ is the unique angle satisfying
\begin{displaymath}
  c_{2k} Q(L_1) = 2 s_{2k} \cos \a_{L_1} \quad \textrm{and} \quad 0 \leq \a_{L_1} \leq \frac{\pi}{2}
\end{displaymath}
and $Q(L_1)$ denotes the function defined by (\ref{qfndef}), continuously extended to the closed interval $[-L_1, L_1]$. Furthermore, the connected component of $\M_{\g}^{-1} \{ (r_1 - r_2,0) \}$ containing $(0,0,r_1,0,0,r_2)$ consists of homoclinic $X_{\g}$-orbits which originate and end in $(0,0,r_1,0,0,r_2)$. Topologically, it is a $2$-dimensional torus, pinched at $(0,0,r_1,0,0,r_2)$.

If $\sqrt{r} < 2 s_{2k} / c_{2k}$, then $(0,0,r_1,0,0,r_2)$ is a \emph{hyperbolic} fixed point of $X_{\g}$ and its stable and unstable manifold have each dimension two. If $\sqrt{r} = 2 s_{2k} / c_{2k}$, $(0,0,r_1,0,0,r_2)$ is an \emph{elliptic} fixed point of $X_{\g}$.
\end{itemize}
\end{prop}

\begin{proof}
(i) By a straightforward computation one shows that under the given assumptions, $(0,0,r_1,0,0,r_2)$ is an elliptic fixed point of $X_{\g}$. In view of equation (\ref{qgeqngck2}) and the discussion of the case $L_1 = \pm r_1$ item (i) then follows easily.

(ii) Again by the discussion preceding Proposition \ref{req1prop} it follows that the inverse image $\M_{\g}^{-1} \{ (r_1 - r_2,0) \}$ is given as claimed. Again by a straightforward computation one sees that $(0,0,r_1,0,0,r_2)$ is a hyperbolic fixed point of $X_{\g}$ if $\sqrt{r} < 2 s_{2k} / c_{2k}$, and an elliptic one if $\sqrt{r} = 2 s_{2k} / c_{2k}$. Next consider a point $(M_i,J_i,L_i)_{1 \leq i \leq 2}$ in $\M_{\g}^{-1} \{ (r_1 - r_2,0) \}$ with
\begin{eqnarray*}
  (M_1,J_1) & = & \sqrt{r_1^2 - L_1^2} \; (\cos \phi, -\sin \phi), \\
  (M_2,J_2) & = & \sqrt{r_2^2 - L_2^2} \; (\cos (\a_{L_1}+\phi), \sin (\a_{L_1}+\phi))
\end{eqnarray*}
where $|L_1| < r_1$. Then the third component of $X_{\g}$ is given by (cf (\ref{req1m1j2m2j1}))
\begin{displaymath}
  (X_{\g})_3 = -\sqrt{r_1^2 - L_1^2} \, \sqrt{r_2^2 - L_2^2} \, \sin \a_{L_1}.
\end{displaymath}
Hence $(X_{\g})_3 = (X_{\g})_6 < 0$. It follows that the $X_{\g}$-orbit passing through such a point originates at $(0,0,r_1,0,0,r_2)$ and then reaches a point of the form $(0,0,-r_1,M_2,J_2,L_2)$ with 
\begin{equation} \label{l2m2j2bottom}
  L_2 = r_2 - 2 r_1 > -r_2 \quad \textrm{and} \quad (M_2,J_2) = \sqrt{r_2^2 - L_2^2} (\cos (\pi + \tilde{\phi}), \sin (\pi + \tilde{\phi})).
\end{equation}
At this point the vector field $X_{\g}$ is given by
\begin{displaymath}
  (-r_1 J_2, -r_1 M_2,0,0,0,0).
\end{displaymath}
Note that this vector does not vanish as $M_2^2 + J_2^2 = r_2^2 - L_2^2 > 0$. Similarly, at a point $(M_i,J_i,L_i)_{1 \leq i \leq 2}$ in $\M_{\g}^{-1} \{ (r_1 - r_2,0) \}$ satisfying
\begin{eqnarray*}
  (M_1,J_1) & = & \sqrt{r_1^2 - L_1^2} \; (\cos \phi, -\sin \phi), \\
(M_2,J_2) & = & \sqrt{r_2^2 - L_2^2} \; (\cos (-\a_{L_1}+\phi), \sin (-\a_{L_1}+\phi))
\end{eqnarray*}
and $|L_1| < r_1$ one has
\begin{displaymath}
  (X_{\g})_3 = \sqrt{r_1^2 - L_1^2} \, \sqrt{r_2^2 - L_2^2} \, \sin \a_{L_1}.
\end{displaymath}
Hence $(X_{\g})_3 = (X_{\g})_6 > 0$. It follows that the $X_{\g}$-orbit passing through such a point ends up at $(0,0,r_1,0,0,r_2)$ and passes through a point of the form $(0,0,-r_1,M_2,J_2,L_2)$ with $(M_2,J_2,L_2)$ as in (\ref{l2m2j2bottom}). We then conclude that the connected component of $\M_{\g}^{-1} \{ (r_1 - r_2,0) \}$ containing $(0,0,r_1,0,0,r_2)$ consists of homoclinic $X_{\g}$-orbits originating and ending at $(0,0,r_1,0,0,r_2)$.
\end{proof}

Finally let us treat the case $r<1$ and $\g + c_k^2 \neq 0$. Denote by $Q(L_1)$ the function defined by (\ref{qfndef}), extended continuously to the closed interval $[-r_1,r_1]$. Further introduce the function
\begin{equation} \label{ffndef}
  f: (0,\sqrt{r}) \to \R, \, q \mapsto (\g + s_k^2) q + (\g + c_k^2) \frac{1}{q}.
\end{equation}
Note that $\lim_{q \searrow 0} |f(q)| = \infty$.

\begin{prop} \label{propgck2neq0}
Assume that $1 \leq k < \fn4$, $0 < r < 1$, and $\g + c_k^2 \neq 0$. Then the following statements hold:
\begin{itemize}
\item[(i)] If $|f(\sqrt{r})| \geq 2 s_{2k}$, 
then the connected component of $\M_{\g}^{-1} \{ (r_1 - r_2,0) \}$ containing the critical point $(0,0,r_1,0,0,r_2)$ consists of this point alone. It is an \emph{elliptic} fixed point of $X_{\g}$.
\item[(ii)] If $|f(\sqrt{r})| < 2 s_{2k}$, 
then there exists $-r_1 < l_{\g,r} < r_1$ so that the connected component of $\M_{\g}^{-1} \{ (r_1 - r_2,0) \}$ containing $(0,0,r_1,0,0,r_2)$ is given by
  \begin{displaymath}
    \bigcup_{l_{\g,r} \leq L_1 \leq r_1 \atop L_2 = L_1 + r_2 - r_1} \mathcal{L}(L_1,L_2,\a_{L_1}) \cup \mathcal{L}(L_1,L_2,-\a_{L_1})
  \end{displaymath}
where for any $l_{\g,r} \leq L_1 \leq r_1$, $\a_{L_1}$ is the unique angle satisfying $0 \leq \a_{L_1} \leq \pi$ and
\begin{displaymath}
  f(Q(L_1)) = -2 s_{2k} \cos \, (\a_{L_1}).
\end{displaymath}
The point $(0,0,r_1,0,0,r_2)$ is a \emph{hyperbolic} fixed point of $X_{\g}$ and its stable and unstable manifold each have dimension two. The connected component of $\M_{\g}^{-1} \{ (r_1 - r_2,0) \}$ containing $(0,0,r_1,0,0,r_2)$ consists of homoclinic $X_{\g}$-orbits which originate and end in $(0,0,r_1,0,0,r_2)$. Topologically, it is a $2$-dimensional torus, pinched at $(0,0,r_1,0,0,r_2)$.
\end{itemize}
\end{prop}

\begin{proof}
(i) By a straightforward computation one shows that under the given assumptions, $(0,0,r_1,0,0,r_2)$ is a (possibly degenerate) elliptic fixed point of $X_{\g}$. We have already seen that under the given assumption $\M_{\g}^{-1} \{ (r_1 - r_2,0) \} \setminus \{ (0,0,r_1,0,0,r_2) \}$ consists of the set of points $(M_i,J_i,L_i)_{1 \leq i \leq 2}$ satisfying $L_2 = L_1 + r_2 - r_1$ and (\ref{qgeqn}). Note that equation (\ref{qgeqn}) admits a solution $\a$ for $Q = \sqrt{r}$ iff $|f(\sqrt{r})| \leq s_{2k}$. In the case $|f(\sqrt{r})| > 2 s_{2k}$ it follows immediately that $\M_{\g}^{-1} \{ (r_1 - r_2,0) \} = \{ (0,0,r_1,0,0,r_2) \}$. If $|f(\sqrt{r})| = 2 s_{2k}$, then an analysis of the graph of $f$ near $(\sqrt{r}, f(\sqrt{r}))$ leads to the claimed result.

(ii) As $\lim_{q \searrow 0} |f(q)| = \infty$ it follows that there exists $-r_1 < l_{\g,r} < r_1$ so that the interval $[l_{\g,r}, r_1]$ is a connected component of $(f \circ Q)^{-1} ([-2 s_{2k}, 2 s_{2k}])$. It follows that for any $l_{\g,r} \leq L_1 \leq r_1$ there exists a unique angle $0 \leq \a_{L_1} \leq \pi$ so that
\begin{displaymath}
  f(Q(L_1)) = -2 s_{2k} \cos \, (\a_{L_1}).
\end{displaymath}
The connected component of the preimage $\M_{\g}^{-1} \{ (r_1 - r_2,0) \}$ containing the point $(0,0,r_1,0,0,r_2)$ is then given as claimed. Again by a straightforward computation one sees that $(0,0,r_1,0,0,r_2)$ is a hyperbolic fixed point of $X_{\g}$. One then can argue as in the proof of item (ii) of Proposition \ref{propgck2eq0} to show the remaining claims.
\end{proof}

\begin{proof}[Proof of Theorem \ref{kgenfolthm}]
Theorem \ref{kgenfolthm} follows from Propositions \ref{req1prop}\,-\,\ref{propgck2neq0}.
\end{proof}

It remains to study the critical points of $\M_{\g}$ with rank $d\M_{\g} = 1$, i.e. points of $(\S_{r_1}^2 \times \S_{r_2}^2) \setminus \{ \pm (0,0,r_1,0,0, \pm r_2) \}$ where the vector fields $Y$ and $X_{\g}$ are collinear. In view of the formulas (\ref{yxgammavf}) for $Y$ and $X_{\g}$, points $(M_i,J_i,L_i) \in S^2_{r_i}$, $i=1,2$, of this type have the property that the determinant of any $2 \times 2$-submatrix of the $2 \times 4$-matrix formed by $Y$ and $X_{\g}$ vanishes. It leads to the following system of equations:
\begin{equation} \label{mjeqn}
  M_1 J_2 + M_2 J_1 = 0, 
\end{equation}
\begin{equation} \label{jleqn}
  J_1^2 L_2 + L_1 J_2^2 - J_1 J_2 (d_{1,\g} L_1 + d_{2,\g} L_2) = 0,
\end{equation}
\begin{equation} \label{mleqn}
  M_1^2 L_2 + L_1 M_2^2 + M_1 M_2 (d_{1,\g} L_1 + d_{2,\g} L_2) = 0,
\end{equation}
\begin{equation} \label{jmleqn}
  M_1 J_1 L_2 - L_1 M_2 J_2 + J_1 M_2 (d_{1,\g} L_1 + d_{2,\g} L_2) = 0.
\end{equation}

\begin{theorem}
Assume that $1 \leq k < \fn4$, $0 < r \leq 1$, and $\g \in \R$. If a point $(M_i,J_i,L_i)_{1 \leq i \leq 2} \in \S_{r_1}^2 \times \S_{r_2}^2 \setminus \{ \pm (0,0,r_1,0,0, \pm r_2) \}$ is a critical point of $\M_{\g}$ with rank $d\M_{\g} = 1$ then $(M_2,L_2) = \l (M_1,-J_1)$ for some $\l \in \R$, and
\begin{displaymath}
  (r_1^2 - L_1^2)^2 L_2^2 + (r_2^2 - L_2^2)^2 L_1^2 + 2 (r_1^2 - L_1^2) (r_2^2 - L_2^2) (2 L_1 L_2 - (d_{1,\g} L_1 + d_{2,\g} L_2)^2) = 0.
\end{displaymath}
Given any point $(M_1,J_1,L_1) \in \S_{r_1}^2 \setminus \{ \pm (0,0,r_1) \}$ there exist at most eight points $(M_2,J_2,L_2) \in \S_{r_2}^2 \setminus \{ \pm (0,0,r_2) \}$ such that $(M_i,J_i,L_i)_{1 \leq i \leq 2}$ is a critical point of $\M_{\g}$ with rank $d\M_{\g} = 1$.
\end{theorem}

\begin{proof}
First assume that $L_1 \in \{ \pm r_1 \}$. Then $J_1 = M_1 = 0$. Hence (\ref{mjeqn}) is automatically satisfied and equations (\ref{jleqn}) and (\ref{mleqn}) read $J_2 = 0$ and $M_2 = 0$, respectively. As a consequence, $(M_1,J_1,L_1) = (0,0, \pm r_1)$ and $(M_2,J_2,L_2) = (0,0, \pm r_2)$. In view of Theorem \ref{kgenfolthm} we thus may assume that $|L_1| < r_1$. Then $(M_1,J_1) \neq (0,0)$. Hence the first equation (\ref{mjeqn}) says that there exists $\l \in \R$ such that
\begin{equation} \label{m2j2m1j2}
  (M_2,J_2) = \l (M_1,-J_1).
\end{equation}
The conditions $(M_i,J_i,L_i) \in \S^2_{r_i}$, $i=1,2$ then imply that $\l$ satisfies
\begin{equation} \label{lambdarlrl}
  \l^2 = \frac{r_2^2 - L_2^2}{r_1^2 - L_1^2}.
\end{equation}
Substituting (\ref{m2j2m1j2}) into (\ref{jleqn})-(\ref{jmleqn}) one sees, again using $(M_1,J_1) \neq (0,0)$, that (\ref{jleqn})-(\ref{jmleqn}) is equivalent to
\begin{equation} \label{l2l1lambdaeqn}
  L_2 + \l^2 L_1 + \l (d_{1,\g} L_1 + d_{2,\g} L_2) = 0,
\end{equation}
or, taking squares, $(L_2 + \l^2 L_1)^2 - \l^2 (d_{1,\g} L_1 + d_{2,\g} L_2)^2 = 0$. Using (\ref{lambdarlrl}), the latter equation reads
\begin{displaymath}
  (r_1^2 - L_1^2)^2 L_2^2 + (r_2^2 - L_2^2)^2 L_1^2 + 2 (r_1^2 - L_1^2) (r_2^2 - L_2^2) (2 L_1 L_2 - (d_{1,\g} L_1 + d_{2,\g} L_2)^2) = 0,
\end{displaymath}
or, after dividing by $r_1^2 r_2^4$ one gets, using the bifurcation parameter $r$ and the normed variables $l_i := L_i/r_i \in (0,1)$ ($i=1,2$),
\begin{equation} \label{lnormedeqn}
  r^2 (1 \!-\! l_1^2)^2 l_2^2 + (1 \!-\! l_2^2)^2 l_1^2 + 2 r (1 \!-\! l_1^2) (1 \!-\! l_2^2) (2 l_1 l_2 \!-\! (\sqrt{r} d_{1,\g} l_1 \!+\! \frac{1}{\sqrt{r}} d_{2,\g} l_2)^2) = 0.
\end{equation}
Note that for given $r$ and $0 < l_1 < 1$, the left hand side of (\ref{lnormedeqn}) is a polynomial in $l_2$ of degree four, $(l_1^2 + d_{2,\g}^2 (1 - l_1^2)) l_2^4 + O(l_2^3)$. 

\begin{figure}
  \centering
  \subfigure[$k=1$]{
     \label{rang1fig:k=1}
     \includegraphics[width=1in]{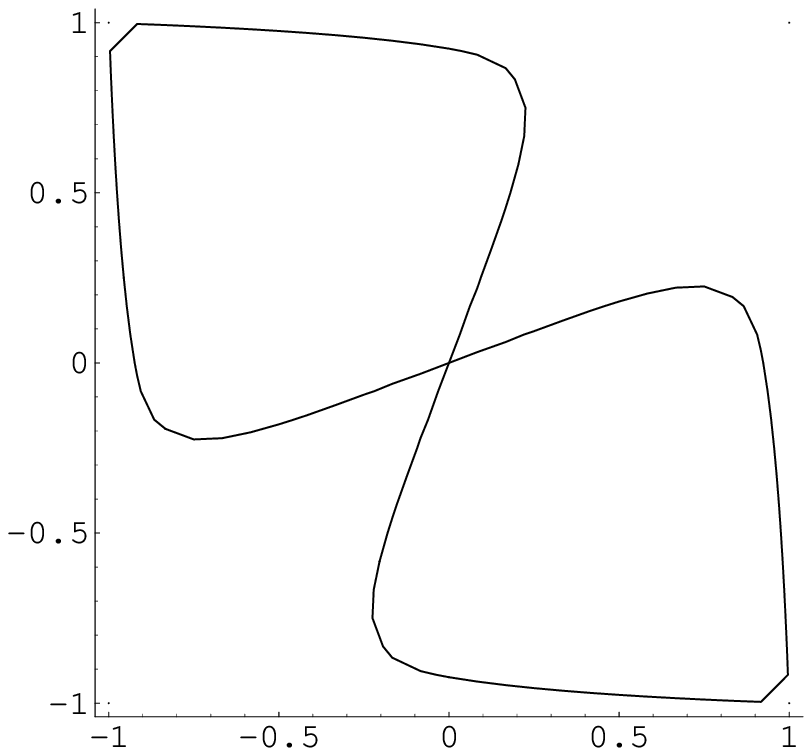}}
  \hspace{.5in}
  \subfigure[$k=3$]{
     \label{rang1fig:k=3}
     \includegraphics[width=1in]{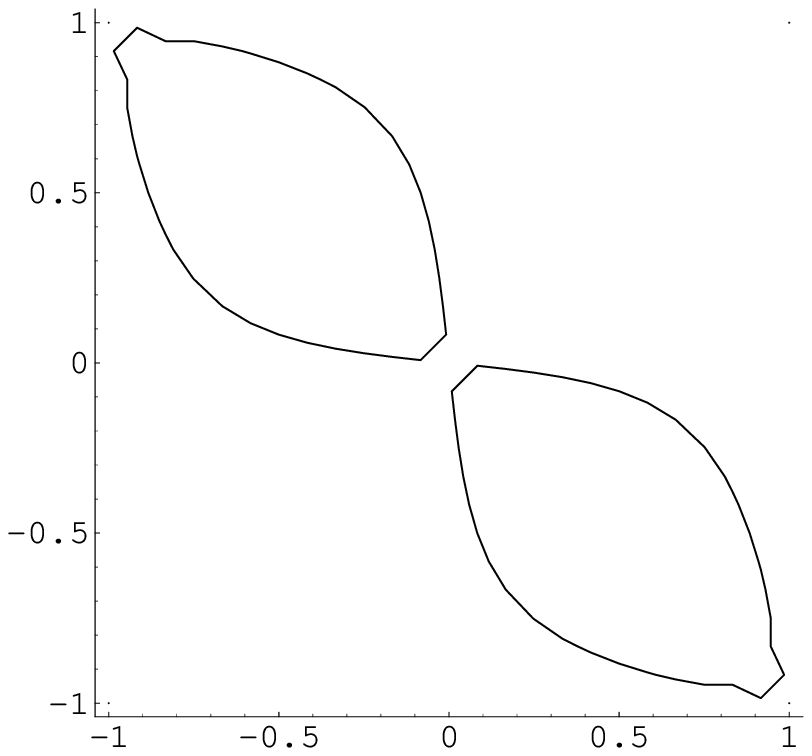}}
  \hspace{.5in}
  \subfigure[$k=5$]{
     \label{rang1fig:k=5}
     \includegraphics[width=1in]{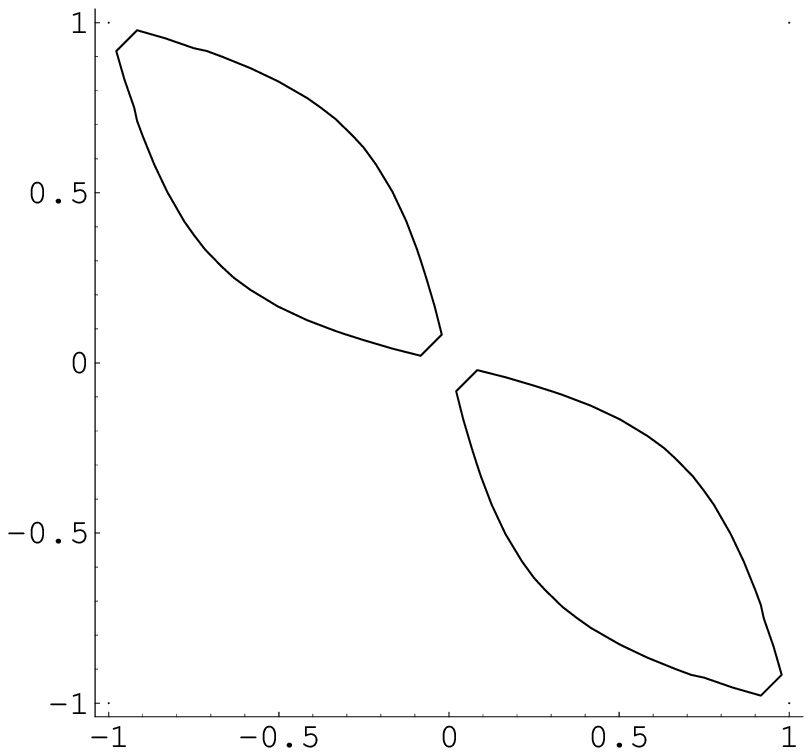}}
  \caption{Sets of solutions $(l_1,l_2)$ of (\ref{lnormedeqn}) for $N=24$, $r=1$, $\g=-0.5$}
  \label{fig:rank1setr=1}
\end{figure}
  

\begin{figure}
  \centering
  \subfigure[$k=1$]{
     \label{rang1r=06fig:k=1}
     \includegraphics[width=1in]{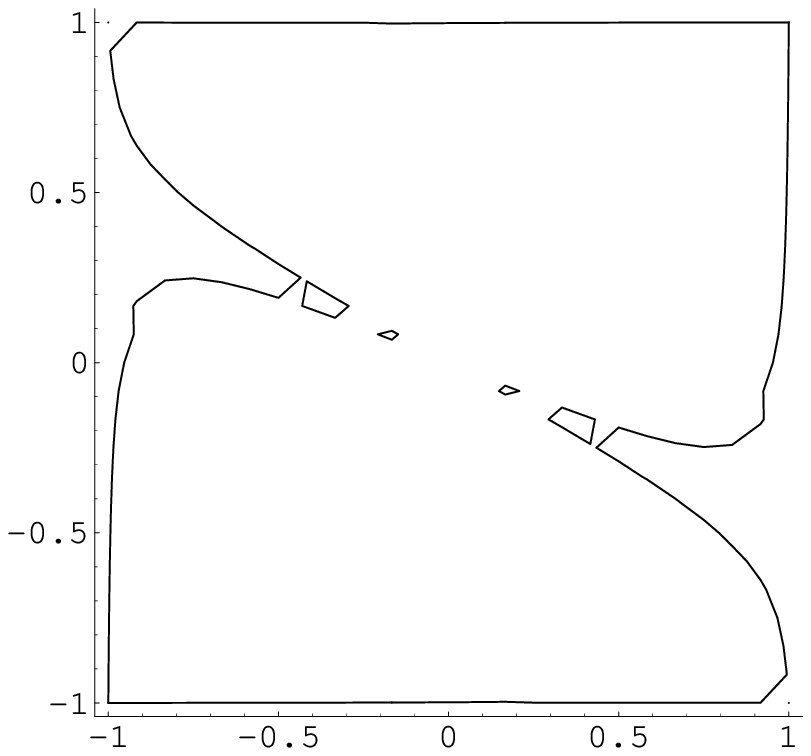}}
  \hspace{.5in}
  \subfigure[$k=3$]{
     \label{rang1r=06fig:k=3}
     \includegraphics[width=1in]{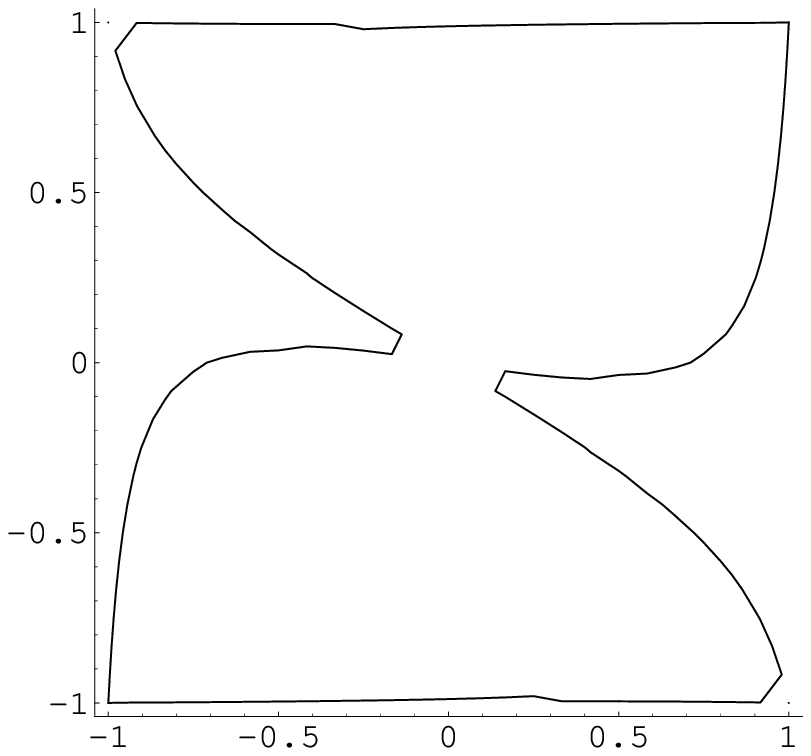}}
  \hspace{.5in}
  \subfigure[$k=5$]{
     \label{rang1r=06fig:k=5}
     \includegraphics[width=1in]{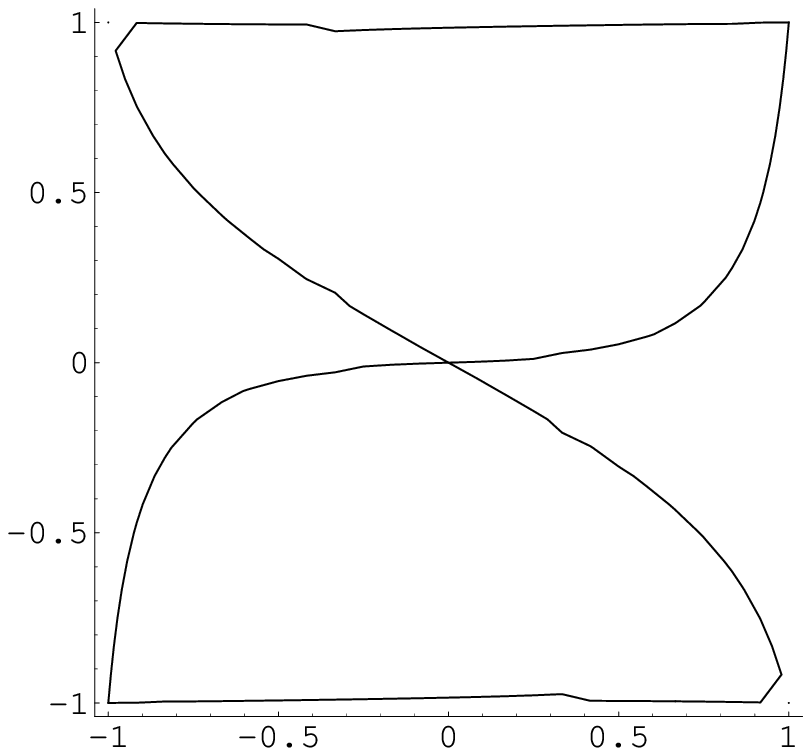}}
  \caption{Sets of solutions $(l_1,l_2)$ of (\ref{lnormedeqn}) for $N=24$, $r=0.3$, $\g=-2.5$}
  \label{fig:rank1setr=06}
\end{figure}

Summarizing, we have shown that for any given point $(M_1,J_1,L_1) \in \S_{r_1}^2 \setminus \{ (0,0, \pm r_1) \}$, there exist at most $8$ points $(M_2,J_2,L_2) \in \S_{r_2}^2 \setminus \{ (0,0, \pm r_2) \}$ such that $Y$ and $X_{\g}$ are collinear. Indeed for any $(M_1,J_1,L_1) \in \S_{r_1}^2 \setminus \{ (0,0, \pm r_1) \}$ a solution $(M_2,J_2,L_2) \in \S_{r_2}^2 \setminus \{ (0,0, \pm r_2) \}$ of (\ref{mjeqn})-(\ref{jmleqn}) is given by $(M_2,J_2) = \l (M_1,J_1)$ and $L_2 = r_2 l_2$ with $\l$ and $l_2$ satisfying (\ref{lambdarlrl}), respectively (\ref{lnormedeqn}). As a consequence the set of solutions of (\ref{mjeqn})-(\ref{jmleqn}) is an algebraic subset of $\S_{r_1}^2 \times \S_{r_2}^2$ of dimension at most two.
\end{proof}

In order to analyze the critical points of $\M_{\g}$ with rank$(d\M_{\g}) = 1$, we perform another symplectic reduction. First we pass to the orbit space of the flow of $Y$ on $\S^2_{r_1} \times \S^2_{r_2}$ and then to the level sets of $G$.

In view of (\ref{orbitsgaction}), the $Y$-flow is an $\S^1$-action. 
Note that besides $L_1$ and $L_2$, the quantities $\s$ and $\t$ are invariant under this $\S^1$-action,
  \[ \s := M_1 M_2 - J_1 J_2, \quad \t := M_1 J_2 + M_2 J_1.
\]
They are related by the identity
\begin{equation} \label{sigmataul1l2}
  \s^2 + \t^2 = \prod_{i=1}^2 (r_i^2 - L_i^2).
\end{equation}

Define
\begin{eqnarray*}
  \mathcal{F}^{(3)}: & \dot{\S}^2_{r_1} \times \dot{\S}^2_{r_2} & \to \quad \R^4 \\
  & \left( M_1, J_1, L_1, M_2, J_2, L_2 \right) & \mapsto (L_1, L_2, \s, \t)
\end{eqnarray*}
where $\dot{\S}^2_{r_i} := \S^2_{r_i} \setminus \{ (0,0, \pm r_i) \}$. Let $\mathcal{O}_r$ denote the image of $\mathcal{F}^{(3)}$. For any element $(L_1,L_2,s,t) \in \mathcal{O}_r$ we have
\begin{equation} \label{pr1r2def}
  s^2 + t^2 = \prod_{i=1}^2 (r_i^2 - L_i^2) \quad \textrm{and} \quad |L_i| \leq r_i \; (i=1,2).
\end{equation} 
The fibers of $\mathcal{F}^{(3)}$ are the orbits of the $Y$-action on $\dot{\S}^2_{r_1} \times \dot{\S}^2_{r_2}$, i.e. $\mathcal{O}_r$ coincides with the orbit space of the $Y$-action on $\dot{\S}^2_{r_1} \times \dot{\S}^2_{r_2}$. As a consequence, any function on $\dot{\S}^2_{r_1} \times \dot{\S}^2_{r_2}$ which Poisson commutes with $G$ factors through $\mathcal{O}_r$.

In particular, $K_{\g}$ and $G$ factor through $\mathcal{O}_r$. In fact, $K_{\g}$ and $G$, when expressed in the variables $L_1$, $L_2$, $\s$, $\t$, are polynomials, given by
\setlength\arraycolsep{2pt}{
\begin{eqnarray}
  K_{\g} & = & \sum_{i=1}^2 \frac{1}{2} d_{i,\g} (r_i^2 - L_i^2) + \s, \label{kkexpr1}\\
  G & = & L_1 - L_2. \label{lkexpr}
\end{eqnarray}}
By reducing the system $(G, K_{\g})$ by the $Y$-action one obtains a family of integrable systems with one degree of freedom parametrized by the value $c$ of $G$. Denote by $X_{\g,c}$ the Hamiltonian vector field induced by $K_{\g}$. The fixed points of $X_{\g,c}$ can then be characterized in terms of the bifurcation parameters $\g$, $r$, and $k$.

Note that by (\ref{mjeqn}), the rank-$1$-points of the reduced moment map $\M_{\g}$ satisfy $\t=0$, and by (\ref{m2j2m1j2})-(\ref{lambdarlrl}), $\s^2 = (r_2^2 - L_2^2)(r_1^2 - L_1^2)$. Hence the image of the set of the rank-$1$-points by $\mathcal{F}^{(3)}$ is an algebraic subset of $\mathcal{O}_{r}$ of dimension at most one - see (\ref{lnormedeqn}).

By (\ref{sigmataul1l2}), $\s$ and $\t$ are located on a circle of radius $\sqrt{(r_1^2 - L_1^2)(r_2^2 - L_2^2)}$,
\begin{equation} \label{phidef}
  (\s, \t) = \sqrt{(r_1^2 - L_1^2)(r_2^2 - L_2^2)} \; (\cos \phi, \sin \phi),
\end{equation}
where $\phi \in \R / 2\pi\Z$. The phase spaces, reduced by the $Y$-action, are now obtained by taking subsets of $\mathcal{O}_r$ corresponding to level sets of $G$, i.e. by replacing $L_2$ by $L_1 - c$, where $c$ is the value of $G$. The restriction $K_{\g,c}$ of $K_{\g}$ to the reduced phase space corresponding to the value $c$ of $G$ is then given by
\begin{eqnarray}
  K_{\g,c}(L_1,\phi) & = & \frac{1}{2} \left( d_{1,\g} (r_1^2 - L_1^2) + d_{2,\g} (r_2^2 - (L_1 - c)^2) \right) \nonumber\\
&& \qquad + \sqrt{(r_1^2 - L_1^2)(r_2^2 - (L_1 - c)^2)} \; \cos \phi \label{kgammacl1phi}
\end{eqnarray}
with $L_1 \in ((-r_1, r_1) \cap (c-r_2, c+r_2))$ and $\phi \in \R / 2\pi\Z$.

The reduced Hamiltonian vector field induced by $K_{\g,c}$ is given by
\begin{equation} \label{ddtl1eqnvect}
    X_{\g,c}(L_1,\phi) = \frac{d}{dt} \left( \begin{array}{c} L_1 \\ \phi \end{array} \right) = \{ L_1, \phi \} \left( \begin{array}{c} \partial K_{\g,c} / \partial \phi \\ -\partial K_{\g,c} / \partial L_1 \end{array} \right).
\end{equation}
Note that
\begin{eqnarray*}
\{ L_1, \phi \} & = & \left\{ L_1, \arctan \frac{\t}{\s} \right\} = \frac{1}{1 + (\t/\s)^2} \left\{ L_1, \frac{\t}{\s} \right\} \\
& = & \frac{1}{1 + (\t/\s)^2} \cdot \frac{\s \{ L_1, \t \} - \t \{ L_1, \s \}}{\s^2} = \frac{1}{1 + (\t/\s)^2} \cdot \frac{\s^2 + \t^2}{\s^2} = 1,
\end{eqnarray*}
since $\{ L_1, \t \} = \s$ and $\{ L_1, \s \} = -\t$. Furthermore, with $L_2 = L_1 - c$,
\begin{eqnarray*}
  \frac{\partial K_{\g,c}}{\partial \phi} & = & - \sqrt{(r_1^2 - L_1^2)(r_2^2 - L_2^2)} \; \sin \phi, \\
\frac{\partial K_{\g,c}}{\partial L_1} & = & -(d_{1,\g} L_1 + d_{2,\g} L_2) - \frac{L_1 (r_2^2 - L_2^2) + L_2 (r_1^2 - L_1^2)}{\sqrt{(r_1^2 - L_1^2)(r_2^2 - L_2^2)}} \cos \phi.
\end{eqnarray*}
Hence (\ref{ddtl1eqnvect}) reads
\begin{equation} \label{ddtl1expl}
    X_{\g,c}(L_1,\phi) = \left( \begin{array}{c} -\sqrt{(r_1^2 - L_1^2)(r_2^2 - L_2^2)} \; \sin \phi \\ (d_{1,\g} L_1 + d_{2,\g} L_2) + \frac{L_1 (r_2^2 - L_2^2) + L_2 (r_1^2 - L_1^2)}{\sqrt{(r_1^2 - L_1^2)(r_2^2 - L_2^2)}} \, \cos \phi \end{array} \right),
\end{equation}
where we treat $L_2 = L_1 - c$ as a dependent variable.

By (\ref{ddtl1expl}), the fixed points of the vector field $X_{\g,c}$ with $|L_1| < r_1$ are given by $(L_1,\phi)$ satisfying
\begin{equation} \label{phieqn}
  \phi \in \pi\Z 
\end{equation}
and
\begin{equation} \label{L1L2eqn}
(d_{1,\g} L_1 + d_{2,\g} L_2) + \frac{L_1 (r_2^2 - L_2^2) + L_2 (r_1^2 - L_1^2)}{\sqrt{(r_1^2 - L_1^2)(r_2^2 - L_2^2)}} \cos \phi = 0.
\end{equation}
Note that (\ref{phieqn}) and the square of (\ref{L1L2eqn}) are equivalent to the system of equations (\ref{mjeqn})-(\ref{lnormedeqn}) derived above.

In order to determine the type of the fixed points $(L_1, \phi)$, i.e. points satisfying (\ref{phieqn})-(\ref{L1L2eqn}) for a given value $c$ of $G$, one computes the Jacobian of $X_{\g,c}$, $H \equiv H_{\g,c}(L_1,\phi) = \left( \begin{array}{cc} h_{11} & h_{12} \\ h_{21} & h_{22} \end{array} \right)$ at these points. Note that at such points $h_{11}=0$ and $h_{22}=0$ and thus $\det(H) = -h_{12} h_{21}$. Hence such a fixed point is an elliptic or hyperbolic fixed point of $X_{\g,c}$ if $h_{12} h_{21}$ is positive, respectively negative. We omit a more detailed analysis of these points. 


\section{Chains with Dirichlet boundary conditions} \label{fpudir}

In this section we consider a FPU chain with $N'$ ($N' \geq 3$, not necessarily even) moving particles and fixed endpoints, i.e. with boundary conditions (\ref{dirboundcond}).

It has been observed that such a chain can be treated as an invariant subsytem of a periodic lattice with $N = 2N'+2$ particles - see \cite{rink06}: Let $T^*\R^N$ be endowed with the canonical symplectic structure and consider the linear map $S: T^*\R^N \to T^*\R^N$ mapping $(q_i)_{1 \leq i \leq N}, (p_i)_{1 \leq i \leq N}$ to
        \[ -(q_{N-1}, \ldots, q_1, q_N), -(p_{N-1}, \ldots, p_1, p_N).
\]
Then $S$ is a canonical linear involution satisfying $H_V \circ S = H_V$. Denote by Fix$(S)$ the fixed point set of $S$. Then Fix$(S)$ is the subset of all elements $(q,p)$ in $T^*\R^N$ satisfying
\begin{eqnarray}
  (q_n, p_n) = -(q_{N-n}, p_{N-n}) \;\; \forall \, 1 \leq n \leq N-1 \; \textrm{and} \; q_N = p_N = 0. \label{fixsdef}
\end{eqnarray}
In particular, on Fix$(S)$, $q_N = q_{N'+1} = 0$ and $p_N = p_{N'+1} = 0$. Note that on Fix$(S)$, both the center of mass coordinate $Q = \frac{1}{N} \sum_{i=1}^N q_i$ and its momentum $P = \frac{1}{N} \sum_{i=1}^N p_i$ are identically $0$. Hence Fix$(S) \subseteq \mathcal{M}$, where
        \[ \mathcal{M} := \{ (q,p) \in T^*\R^N | Q=0; P=0 \}.
\]
We endow $\mathcal{M}$ with the symplectic structure induced from $T^*\R^N$.

The phase space of an FPU chain with $N'$ moving particles satisfying Dirichlet boundary conditions is $T^*\R^{N'}$, endowed with the canonical symplectic structure $\sum_{i=1}^{N'} dq_i \wedge dp_i$. It can be embedded into $\mathcal{M}$ by the map $\Theta: T^*\R^{N'} \to \mathcal{M}$ defined by
\begin{displaymath}
        (q_i,p_i)_{1 \! \leq \! i \! \leq \! N'} \mapsto \frac{1}{\sqrt{2}} \left( (q_i,p_i)_{1 \leq i \leq N'}, (0, 0), -(q_{N'\!-\!i},p_{N'\!-\!i})_{0 \leq i \leq N'\!-\!1}, (0, \! 0) \! \right).
\end{displaymath}
Note that $\Theta(T^*\R^{N'}) = Fix(S)$, i.e. $\Theta$ is a parametrization of Fix$(S)$ and the pullback of the canonical symplectic form on $\mathcal{M}$ by $\Theta$ is $\sum_{i=1}^{N'} dq_i \wedge dp_i$, which means that $\Theta$ is canonical. It then follows that Fix$(S)$ is a symplectic submanifold of $\mathcal{M}$.

We now express the equations defining Fix$(S)$ locally near $0$ as a subset of $\mathcal{M}$ in terms of the canonical coordinates $(x_k, y_k)_{\1N1}$ provided by Theorem \ref{bnffputheoremeven}, or even more conveniently, in terms of the associated complex coordinates $(\z_k)_{1 \leq |k| \leq N-1}$, defined for $\1N1$ by 
\begin{equation} \label{complexdef}
\left\{ \begin{array}{rcl}
\z_k & = & \frac{1}{\sqrt{2}} (x_k - i y_k) \\
\z_{-k} = \overline{\z_k} & = & \frac{1}{\sqrt{2}} (x_k + i y_k).
\end{array} \right.
\end{equation}
Denote by $\mathcal{Z}$ the linear subspace of $\C^{2N-2}$ consisting of such vectors $(\z_k)_{1 \! \leq \! |k| \! \leq \! N\!-\!1}$. In the sequel we also write $(\z_k)_{\1N1}$ for the element $(\z_k)_{1 \leq |k| \leq N-1} \in \mathcal{Z}$ and use the notations ($n \in \Z$)
\begin{displaymath}
  c_n := \cos \frac{n \pi}{N}, \quad s_n := \sin \frac{n \pi}{N}.
\end{displaymath}

Define the map $\sz: \mathcal{Z} \to \mathcal{Z}$, given by
\begin{equation} \label{szdef}
 (\z_k)_{\1N1} \mapsto (-e^{4 \pi i k / N} \z_{N-k})_{\1N1}.
\end{equation}
Like the map $S: \mathcal{M} \to \mathcal{M}$, $\sz$ is a canonical linear involution. In fact, the maps $S$ and $\sz$ are conjugate to each other under the coordinate change of Theorem \ref{bnffputheoremeven}. Before making this statement more precise, let us introduce a parametrization of the fixed point set Fix$(\sz)$ of the map $\sz$. Introduce
        \[ \zdir := \{ (\z_k)_{1 \leq |k| \leq N'} \in \C^{2N'} | \, \overline{\z}_k = \z_{-k} \quad \forall \, 1 \leq k \leq N' \},
\]
endowed with the \emph{canonical} symplectic structure 
induced from $\C^{2N'}$, and the embedding $\thmz: \zdir \to \mathcal{Z}$ mapping $(\z_k)_{1 \leq |k| \leq N'}$ to the element $(\tilde{\z}_k)_{1 \leq k \leq N} \in \mathcal{Z}$ given by
        \[ \frac{1}{\sqrt{2}} \left( (\z_k)_{1 \leq k \leq N'}, 0, (-e^{4 \pi i k / N} \z_{N'+1-k})_{1 \leq k \leq N'} \right).
\]
Note that $\thmz(\zdir) = Fix(\sz)$, i.e. $\thmz$ is a parametrization of Fix$(\sz)$.

Using the explicit construction in \cite{ahtk4} of the canonical transformation $\Psi$ of Theorem \ref{bnffputheoremeven} one verifies the following lemma.
\begin{lemma} \label{xkxnkproof}
In terms of the complex variables $(\z_k)_{1 \leq |k| \leq N-1}$ defined by Theorem \ref{bnffputheoremeven}, near $0$, the map $S$ is given by $S_\mathcal{Z}$. More precisely, if $\Psi$, defined near $0 \in \mathcal{Z}$, is the coordinate transformation given by Theorem \ref{bnffputheoremeven}, then $S \circ \Psi = \Psi \circ \sz$. In particular, locally near $0$, the set Fix$(S_\mathcal{Z}) \subseteq \mathcal{Z}$, described by the equations
\begin{equation} \label{fixscondzeta}
  e^{-2 \pi i k / N} \z_k + e^{2 \pi i k / N} \z_{N-k} = 0 \quad (\1N1),
\end{equation}
is the image of the set Fix$(S)$ under $\Psi^{-1}$. Expressed in terms of the real variables $(x_k, y_k)_{\1N1}$, the conditions (\ref{fixscondzeta}) are given by
\begin{equation} \label{fixscondxy}
\left( \begin{array}{cc} c_{2k} & -s_{2k} \\ s_{2k} & c_{2k} \end{array} \right) \left( \begin{array}{c} x_k \\ y_k \end{array} \right) + \left( \begin{array}{cc} c_{2k} & s_{2k} \\ -s_{2k} & c_{2k} \end{array} \right) \left( \begin{array}{c} x_{N-k} \\ y_{N-k} \end{array} \right) = \left( \begin{array}{c} 0 \\ 0 \end{array} \right).
\end{equation}
In particular, for $k = N'+1 \, (= N/2)$ we get $\z_{N'+1} = 0$ and therefore
\begin{displaymath}
(x_{N'+1}, y_{N'+1}) = (0,0).
\end{displaymath}
\end{lemma}


\begin{cor} \label{fixsikinminusk}
On Fix$(S_\mathcal{Z})$, for any $1 \leq k \leq \frac{N}{2}$,
\begin{equation} \label{ikeqinminusk}
  I_k = I_{N-k}
\end{equation}
and
\begin{equation} \label{jjmmeqii}
  J_k J_{\frac{N}{2}-k} - M_k M_{\frac{N}{2}-k} = I_k I_{\frac{N}{2}-k}.
\end{equation}
Moreover
\begin{equation} \label{in2eqzero}
  I_{\frac{N}{2}} = 0.
\end{equation}
\end{cor}

\begin{proof}
In terms of the complex variables $(\z_k)_{1 \leq |k| \leq N-1}$, $I_k = \z_k \z_{-k}$ for any $\1N1$. Hence on Fix$(S_\mathcal{Z})$,
  \[ I_k = \z_k \z_{-k} = (-e^{4 \pi i k / N} \z_{N-k})(-e^{-4 \pi i k / N} \z_{-(N-k)}) = \z_{N-k} \z_{-(N-k)} = I_{N-k},
\]
showing (\ref{ikeqinminusk}). The identity (\ref{in2eqzero}) follows from $\z_\frac{N}{2}|_{Fix(S_\mathcal{Z})} = 0$. To prove (\ref{jjmmeqii}), we first conclude from (\ref{fixscondzeta}) that on Fix$(S_\mathcal{Z})$, for any $\1N1$,
\begin{eqnarray*}
 J_k & = & -c_{4k} I_k, \\
 M_k & = & -s_{4k} I_k.
\end{eqnarray*}
Hence on Fix$(S_\mathcal{Z})$,
\begin{eqnarray*}
  J_k J_{\frac{N}{2}-k} - M_k M_{\frac{N}{2}-k} & = & I_k I_{\frac{N}{2}-k} (c_{4k} c_{4(\frac{N}{2}-k)} - s_{4k} s_{4(\frac{N}{2}-k)}) \\
  & = & I_k I_{\frac{N}{2}-k} (c_{4k}^2 + s_{4k}^2) \\
  & = & I_k I_{\frac{N}{2}-k}.
\end{eqnarray*}
This completes the proof of Corollary \ref{fixsikinminusk}.
\end{proof}

From the definitions (\ref{jkmkdef}), (\ref{rdef}), and (\ref{rn4def}) of the variables $I_k$, $J_k$, $M_k$, and of the expressions $R$ and $R_\frac{N}{4}$ one then obtains the following

\begin{cor} \label{dirvarcor}
On Fix$(S_\mathcal{Z})$,
\begin{displaymath}
R = 4 \sum_{1 \leq k < \frac{N}{4}} s_{2k} I_k I_{\frac{N}{2} - k} \quad \textrm{and} \quad R_\frac{N}{4} = \left\{ \begin{array}{ll}
I_\frac{N}{4}^2 & \quad \textrm{if } \frac{N}{4} \in \N \\
0 & \quad \textrm{otherwise.} \end{array} \right.
\end{displaymath}
\end{cor}

It follows from Corollary \ref{dirvarcor} that on Fix$(S_\mathcal{Z})$, the expression (\ref{evenfpuformula}) is in Birkhoff normal form up to order $4$. This allows us to prove Theorem \ref{bnfdirichlettheorem}.

\begin{proof}[Proof of Theorem \ref{bnfdirichlettheorem}]
We start with the resonant normal form (\ref{evenfpuformula}) for even chains, $\frac{N P^2}{2} + \Hab(I) - R_{\a,\b}(J,M) + O(|(x,y)|^5)$, where $\Hab(I)$ and $R_{\a,\b}(J,M)$ are given by (\ref{bnfintrotheorem}) and (\ref{rabdef}), respectively. Using the identity $I_k = I_{N-k}$, the terms in the decomposition (\ref{habintrepr}) of $\Hab$, when restricted to Fix$(S_\mathcal{Z})$, are given by
\begin{equation} \label{h2fixs}
H^{(2)}(I) = 4 \sum_{k=1}^{N'} s_k \, I_k,
\end{equation}
\begin{equation} \label{h4fixs}
H^{(4)}_{\a,\b}(I) = \frac{1}{N} \sum_{k=1}^{N'} d_k^+ I_k^2 + \frac{4(\b - \a^2)}{2N} \sum_{1 \leq k,l \leq N' \atop k \neq l} s_k s_l I_k I_l, 
\end{equation}
and
\begin{equation} \label{ikinkfixs}
\frac{1}{2N} \sum_{k=1}^{\frac{N}{2}-1} d_k^- I_k I_{N-k} = \frac{1}{2N} \sum_{k=1}^{N'} d_k^- I_k^2.
\end{equation}
From Corollary \ref{dirvarcor}, we conclude that on Fix$(S_\mathcal{Z})$,
\begin{eqnarray}
  -R_{\a,\b}(J,M) & = & -\frac{\b - \a^2}{4N} \left( R(J,M) + R_\frac{N}{4}(J,M) \right) \nonumber\\
 & = & -\frac{\b - \a^2}{4N} \left( 4 \sum_{1 \leq k < \frac{N}{4}} s_{2k} I_k I_{\frac{N}{2} - k} \underbrace{ + \quad I_\frac{N}{4}^2}_{\textrm{only if }\frac{N}{4} \in \N} \right). \label{rabfixs}
\end{eqnarray}
Formula (\ref{bnfdirichletformula}) is then obtained by adding up (\ref{h2fixs})-(\ref{rabfixs}), noting that $d_k^+ + \frac{d_k^-}{2} = \frac{1}{2}(\a^2 + 3(\b-\a^2)s_k^2)$, and replacing $I_k$ by its pullback $\frac{1}{2} I_k$ with respect to the parametrization $\thmz$ of Fix$(S_\mathcal{Z})$ introduced above.
\end{proof}

It remains to prove Theorem \ref{bnfdirprop}. We first consider the case $\a=0$.

\begin{prop} \label{bchaindir}
Assume that $\a = 0$ in (\ref{potentialdef}). Then the following holds:
  \begin{itemize}
  \item[(i)] The Birkhoff normal form of $H_V$ with Dirichlet boundary conditions up to order $4$ is given by $\frac{(N'+1) P^2}{2} + \H0bd(I)$ where
    \begin{eqnarray}
        \H0bd(I) & = & 2 \sum_{k=1}^{N'} s_k I_k + \frac{\b}{16(N'+1)} \Bigg( \sum_{k=1}^{N'} 3 s_k^2 I_k^2 \; \underbrace{ + \frac{1}{2} I_\frac{N'+1}{2}^2}_{\textrm{only if }\frac{N'+1}{2} \in \N} \nonumber\\
        && + 4 \!\! \sum_{l \neq m \atop 1 \leq l,m \leq N'} \!\! s_l s_m I_l I_m - \sum_{k=1}^{N'} s_{2k} I_k I_{N'+1-k} \Bigg). \label{bnfdirbchain}
    \end{eqnarray}
  \item[(ii)] For any $\b \neq 0$, $\H0bd(I)$ is nondegenerate at $I=0$.
  \end{itemize}
\end{prop}

\begin{proof}
The Birkhoff normal form (\ref{bnfdirbchain}) of $H_V$ with Dirichlet boundary conditions is given by the formula (\ref{bnfdirichletformula}) evaluated at $\a=0$. To investigate the Hessian of $\Qbd$ of $\H0bd(I)$ at $I=0$, we write
\begin{equation} \label{q0bddecomp}
 \Qbd = \frac{2\b}{16(N'+1)} \Delta^{N'} P^D \Delta^{N'},
\end{equation}
where $\Delta^{N'} = \textrm{diag} \, \left( \sin\frac{k\pi}{2N'+2} \right)_{1 \leq k \leq N'}$ and $P^D$ is the $N' \times N'$-matrix which for $N'$ even resp. odd is of the form
\begin{displaymath}
\underbrace{\left( \begin{array}{cccccccccc}
3 & 4 & \ldots &&&&& \ldots & 4 & 2 \\
4 & 3 & 4 & \ldots &&& \ldots & 4 & 2 & 4 \\
\vdots && \ddots &&&&& \iddots && \vdots \\
&&& 3 & 4 & 4 & 2 &&& \\
&&& 4 & 3 & 2 & 4 &&& \\
&&& 4 & 2 & 3 & 4 &&& \\
&&& 2 & 4 & 4 & 3 &&& \\
\vdots && \iddots &&&&& \ddots && \vdots \\
4 & 2 & 4 & \ldots &&& \ldots & 4 & 3 & 4 \\
2 & 4 & \ldots &&&&& \ldots & 4 & 3
\end{array} \right)}_{(N' \; \textrm{even})}, \quad \underbrace{\left( \begin{array}{ccccccccc}
3 & 4 & \ldots &&&& \ldots & 4 & 2 \\
4 & 3 & 4 & \ldots && \ldots & 4 & 2 & 4 \\
\vdots && \ddots &&&& \iddots && \vdots \\
 &&& 3 & 4 & 2 &&& \\
 &&& 4 & 2 & 4 &&&  \\
 &&& 2 & 4 & 3 &&& \\
\vdots && \iddots &&&& \ddots && \vdots \\
4 & 2 & 4 & \ldots && \ldots & 4 & 3 & 4 \\
2 & 4 & \ldots &&&& \ldots & 4 & 3
\end{array} \right)}_{(N' \; \textrm{odd})},
\end{displaymath}
where we used that $s_{2k} = 2 s_k c_k = 2 s_k s_{N'+1-k}$ and, if $\frac{N'+1}{2} \in \N$, $s_{\frac{N'+1}{2}}^2 = \frac{1}{2}$. It follows that
        \[ \det \Qbd = \left( \frac{2\b}{16(N'+1)} \right)^{N-1} \cdot \det P^D \cdot \prod_{k=1}^{N'} \sin^2 \frac{k\pi}{2N'+2}.
\]
In order to see that $P^D$ is nonsingular, 
observe that $det P^D \in \Z$. 
For $N'$ even we show that $\det P^D \equiv 1$ mod $2$. Note that in this case the diagonal of $P^D$ consists of $3$'s only. Therefore $\det P^D \equiv 3^{N'} \, \textrm{mod } 2 \equiv 1 \, \textrm{mod } 2$. If $N'$ is odd, the same argument shows that $\det P \equiv 2$ mod $4$.
Hence, if $\b \neq 0$, $\det \Qbd \neq 0$, and the nondegeneracy of the Hessian of $\H0bd(I)$ at $I=0$ follows.
\end{proof}

\begin{lemma} \label{signbchaindir}
If $\b < 0$, then $\Qbd$ has $\ulcorner \frac{N'+1}{2} \urcorner$ negative eigenvalues, whereas if $\b > 0$, then $\Qbd$ has $\llcorner \frac{N'-1}{2} \lrcorner$ negative eigenvalues. In particular, for any $\b \neq 0$, $\Qbd$ is indefinite (and $\H0bd$ is therefore not convex).
\end{lemma}

\begin{proof}
We want to use the decomposition (\ref{q0bddecomp}) of $\Qbd$ to show that $\Qbd$ can be deformed continuously to $\frac{2\b}{16(N'+1)} P^D$: Consider for $0 \leq t \leq 1$
\begin{displaymath}
  \Qbd(t) := \frac{2\b}{16(N'+1)} (t \, \Delta^{N'} + (1-t) \, \textrm{Id}) \; P^D \; (t \, \Delta^{N'} + (1-t) \, \textrm{Id}).
\end{displaymath}
As $t \, \Delta^{N'} + (1-t) \, \textrm{Id}$ is positive definite for any $0 \leq t \leq 1$ and $P^D$ is regular and symmetric, $\Qbd(t)$ is a symmetric regular $N' \times N'$-matrix for any $0 \leq t \leq 1$. For $t=0$, $\Qbd(0) = \frac{2\b}{16(N'+1)} P^D$, whereas for $t=1$, $\Qbd(1) = \Qbd$. Therefore, index$(\Qbd)$ (i.e. the number of negative eigenvalues of $\Qbd$) coincides with index$(\frac{2\b}{16(N'+1)} P^D)$. To describe the spectrum of $P^D$, we distinguish between $N'$ even and odd.

If $N'$ is even, the eigenvalues of $P^D$ are $4N'-3$ (with multiplicity one), $1$ (with multiplicity $\frac{N'}{2}$), and $-3$ (with multiplicity $\frac{N'}{2}-1$), hence $P^D$ has $\frac{N'}{2}-1$ negative eigenvalues. If $N'$ is odd, the eigenvalues of $P^D$ are $1$ (with multiplicity $\frac{N'-1}{2}$), $-3$ (with multiplicity $\frac{N'-3}{2}$), and $\frac{1}{2}(4N'-5) \left( 1 \pm \sqrt{1 + \frac{8(4N'-1)}{(4N'-5)^2}} \right)$ (each with multiplicity one), hence $P^D$ has $\frac{N'-1}{2}$ negative eigenvalues. These facts are verified in Appendix \ref{pdeigenvcomp}. The claim of the lemma now follows.
\end{proof}

We now turn to the case $\a \neq 0$.
\begin{prop} \label{dirchaingenproperties}
Assume that $\a \neq 0$ in (\ref{potentialdef}). Then, for $\a$ fixed, $\det \Qabd$ is a polynomial in $\b$ of degree $N'$ and has $N'$ real zeroes (counted with multiplicities). When denoted by $\b_k = \b_k(\a)$ ($1 \leq k \leq N'$) and listed in increasing order, they satisfy
        \[ \b_1 \leq \ldots \leq \b_{\ulcorner \frac{N'+1}{2} \urcorner} < \a^2 < \b_{\ulcorner \frac{N'+3}{2} \urcorner} \leq \ldots \leq \b_{N'}.
\]
Moreover
\begin{displaymath}
  \textrm{index} \, (\Qabd) = \left\{  \begin{array}{ll}
\ulcorner \frac{N'+1}{2} \urcorner & \quad \textrm{for } \b < \b_1 \\
0 & \quad \textrm{for } \b_{\ulcorner \frac{N'+1}{2} \urcorner} < \b < \b_{\ulcorner \frac{N'+3}{2} \urcorner} \\
\llcorner \frac{N'-1}{2} \lrcorner & \quad \textrm{for } \b > \b_{N'}
\end{array} \right.
\end{displaymath}
\end{prop}

\begin{proof}
Fix $\a \in \R \setminus \{ 0 \}$ and consider the map $\b \mapsto \det (\Qabd)$. It follows from (\ref{bnfdirichletformula}) that $\det (\Qabd)$ is a polynomial in $\b$ of degree at most $N'$,
\begin{displaymath}
  \det (\Qabd) = \sum_{j=0}^{N'} r_j \b^j,
\end{displaymath}
where $r_0 = \det(\Qad)$ and $r_{N'} = \det(\Qd)$. By Proposition \ref{bchaindir}, $\det (\Qd) \neq 0$, hence the degree of the polynomial $\det (\Qabd)$ is $N'$. We claim that $\det(\Qabd)$ has $N'$ real zeroes (counted with multiplicities). For $|\b|$ large enough, index$(\Qabd)$ is equal to index$(\Qbd)$. By Lemma \ref{signbchaindir}, index$(\Qbd)$ is $\llcorner \frac{N'-1}{2} \lrcorner$ for $\b > 0$ and $\ulcorner \frac{N'+1}{2} \urcorner$ for $\b < 0$. Hence there exists $R > 0$ such that index$(\Qabd) = \llcorner \frac{N'-1}{2} \lrcorner$ for any $\b > R$ and index$(\Qabd) = \ulcorner \frac{N'+1}{2} \urcorner$ for any $\b < -R$. For $\b = \a^2$, $\Qtoda$ is a positive multiple of the identity matrix, hence index$(\Qtoda) = 0$. It then follows that, when counted with multiplicities, index$(\Qabd)$ must change at least $\ulcorner \frac{N'+1}{2} \urcorner$ times in the open interval $(-\infty, \a^2)$ and at least $\llcorner \frac{N'-1}{2} \lrcorner$ times in $(\a^2, \infty)$. Since a change of index$(\Qabd)$ induces a real zero of $\det (\Qabd)$, our consideration shows that $\b \mapsto \det (\Qabd)$ has $N'$ real zeroes. Further we have $\b_{\ulcorner \frac{N'+1}{2} \urcorner}(\a) < \a^2 < \b_{\ulcorner \frac{N'+3}{2} \urcorner}(\a)$. This proves the proposition.
\end{proof}

\begin{proof}[Proof of Theorem \ref{bnfdirprop}]
Part (i) is proved by Proposition \ref{dirchaingenproperties}, whereas (ii) follows from Propostion \ref{bchaindir} and Lemma \ref{signbchaindir}.
\end{proof}

\appendix





\section{Spectrum of the matrix $P^D$} \label{pdeigenvcomp}

Here we compute for any integer $N' \geq 3$ the eigenvalues of the $N' \times N'$ matrix $P^D$, given by
\begin{displaymath}
\underbrace{\left( \begin{array}{cccccccccc}
3 & 4 & \ldots &&&&& \ldots & 4 & 2 \\
4 & 3 & 4 & \ldots &&& \ldots & 4 & 2 & 4 \\
\vdots && \ddots &&&&& \iddots && \vdots \\
&&& 3 & 4 & 4 & 2 &&& \\
&&& 4 & 3 & 2 & 4 &&& \\
&&& 4 & 2 & 3 & 4 &&& \\
&&& 2 & 4 & 4 & 3 &&& \\
\vdots && \iddots &&&&& \ddots && \vdots \\
4 & 2 & 4 & \ldots &&& \ldots & 4 & 3 & 4 \\
2 & 4 & \ldots &&&&& \ldots & 4 & 3
\end{array} \right)}_{(N' \; \textrm{even})}, \quad \underbrace{\left( \begin{array}{ccccccccc}
3 & 4 & \ldots &&&& \ldots & 4 & 2 \\
4 & 3 & 4 & \ldots && \ldots & 4 & 2 & 4 \\
\vdots && \ddots &&&& \iddots && \vdots \\
 &&& 3 & 4 & 2 &&& \\
 &&& 4 & 2 & 4 &&&  \\
 &&& 2 & 4 & 3 &&& \\
\vdots && \iddots &&&& \ddots && \vdots \\
4 & 2 & 4 & \ldots && \ldots & 4 & 3 & 4 \\
2 & 4 & \ldots &&&& \ldots & 4 & 3
\end{array} \right)}_{(N' \; \textrm{odd})}.
\end{displaymath}

\begin{lemma} \label{eigenvpdtheorem}
If $N'$ is even, the eigenvalues of $P^D$ are $4N'-3$ (with multiplicity one), $1$ (with multiplicity $\frac{N'}{2}$), and $-3$ (with multiplicity $\frac{N'}{2}-1$). If $N'$ is odd, the eigenvalues of $P^D$ are $1$ (with multiplicity $\frac{N'-1}{2}$), $-3$ (with multiplicity $\frac{N'-3}{2}$), and $\frac{1}{2}(4N'-5) \left( 1 \pm \sqrt{1 + \frac{8(4N'-1)}{(4N'-5)^2}} \right)$ (each with multiplicity one).
\end{lemma}

\begin{proof}
Throughout this proof, antidiag$(a_1, \ldots, a_{N'})$ denotes the ``antidiagonal'' $N' \times N'$-matrix $M$ with $M_{kl} = a_l$ if $k+l=N'+1$ and $M_{kl} = 0$ otherwise.

First consider the case where $N'$ is even. We write $P^D$ in the form
        \[ P^D = \textrm{diag}(-1, \ldots, -1) + \textrm{antidiag}(-2, \ldots, -2) + 4 \cdot 1_{N' \times N'}
\]
and, with $\mu := -1-\l$,
        \[ P^D - \l \textrm{Id} = \underbrace{\textrm{diag}(\mu, \ldots, \mu) + \textrm{antidiag}(-2, \ldots, -2)}_{=: L^{(N')}} + 4 \cdot 1_{N' \times N'}.
\]
Here $1_{N' \times N'}$ denotes the $N' \times N'$-matrix whose entries are all equal to $1$.

We compute $\det(P^D - \l \textrm{Id}) = \det(L^{(N')} + 4 \cdot 1_{N' \times N'})$ by column expansion. Note that in the column expansion of the determinant only those terms contribute which are determinants of matrices containing at most one column consisting of entries all equal to four. We obtain
\begin{equation} \label{detpdexp}
  \det(P^D - \l \textrm{Id}) = \det(L^{(N')}) + \sum_{j=1}^{N'} \det(L_j^{(N')}),
\end{equation}
where $L_j^{(N')}$ is defined as the matrix $L^{(N')}$ with the $j$-th column replaced by the column $4 \cdot 1_{N' \times 1}$. By expansion with respect to the first column and then the last column, the determinant of $L^{(N')}$ can be computed recursively,
        \[ \det(L^{(N')}) = (\mu^2 - 2^2) \det(L^{(N'-2)})
\]
Since $\det(L^{(2)}) = \mu^2 - 4$, it follows by induction that
\begin{equation} \label{lmudet}
        \det(L^{(N')}) = (\mu^2 - 4)^{\frac{N'}{2}}.
\end{equation}
To compute $\det(L_1^{(N')})$, we expand the determinant in the same way and obtain the identity $\det(L_1^{(N')}) = 4(\mu + 2) \det(L^{(N'-2)})$, from which it follows that
\begin{equation} \label{lmu1det}
        \det(L_1^{(N')}) = 4(\mu + 2)(\mu^2 - 4)^{\frac{N'}{2} - 1}.
\end{equation}
Similarly one gets $\det(L_2^{(N')}) = (\mu^2 - 4) \det(L_1^{(N'-2)})$, and thus
\begin{equation} \label{lmu2det}
        \det(L_2^{(N')}) = \det(L_1^{(N')}) = 4(\mu + 2)(\mu^2 - 4)^{\frac{N'}{2} - 1}.
\end{equation}
For any $1 < j < \frac{N'}{2}$, this procedure leads to $\det(L_j^{(N')}) = (\mu^2 - 4) \det(L_{j-1}^{(N'-2)})$ and hence
\begin{equation} \label{lmujdet}
        \det(L_j^{(N')}) = \det(L_1^{(N')}) = 4(\mu + 2)(\mu^2 - 4)^{\frac{N'}{2} - 1}.
\end{equation}
For $j > \frac{N'}{2}$, note that $\det L_j^{(N')} = \det L_{N-j}^{(N')}$, since $L_j^{(N')}$ and $L_{N-j}^{(N')}$ can be transformed into each other by exchanging the $j$'th and the $(N-j+1)$'th columns and then the $j$'th and the $(N-j+1)$'th rows. By (\ref{lmudet})-(\ref{lmujdet}), we obtain
\begin{eqnarray*}
  \det(P^D - \l \textrm{Id}) & = & (\mu^2 - 4)^{\frac{N'}{2} - 1} \cdot ((\mu^2 - 4) + N' \cdot 4(\mu + 2)) \\
  & = & (\mu^2 - 4)^{\frac{N'}{2} - 1} (\mu + 2) (\mu - 2 + 4N').
\end{eqnarray*}
Hence, if $N'$ is \emph{even}, the zeroes of $\det(P^D - \l \textrm{Id})$ are $\mu = 2$ (with multiplicity $\frac{N'}{2} - 1$), $\mu = -2$ (with multiplicity $\frac{N'}{2}$), and $\mu = -4N' + 2$ (with multiplicity $1$). Transforming back to $\l = -1-\mu$, we obtain the claimed eigenvalues.

It remains to consider the case where $N'$ is \emph{odd}. Again, we write
        \[ P^D = \textrm{diag}(-1, \ldots, -1,\overbrace{0}^{(\frac{N'+1}{2})},-1, \ldots, -1) + \textrm{antidiag}(-2, \ldots, -2) + 4 \cdot 1_{N' \times N'}.
\]
With $\mu = -1-\l$ we get
  \[ P^D - \l \textrm{Id} = L^{(N')} + 4 \cdot 1_{N' \times N'}
\]
with
  \[ L^{(N')} = \left( \begin{array}{ccccc}
        \mu & 0 & \ldots & 0 & -2 \\
        0 & \ddots && \iddots & 0 \\
         \vdots && \mu-1 && \vdots \\
        0 & \iddots && \ddots & 0 \\
        -2 & 0 & \ldots & 0 & \mu
        \end{array} \right).
\]

As above, we obtain the expansion (\ref{detpdexp}) for the determinant of $P^D - \l \textrm{Id}$. We expand $\det(L^{(N')})$ with respect to the first column and then the last column, yielding the recursion formula
        \[ \det(L^{(N')}) = (\mu^2 - 4) \det(L^{(N'-2)}),
\]
which together with $\det(L^{(1)}) = \mu - 1$ leads to
\begin{equation} \label{lmudetodd}
        \det(L^{(N')}) = (\mu^2 - 4)^{\frac{N'-1}{2}} (\mu - 1).
\end{equation}
For $\det(L_1^{(N')})$, we obtain the identity $\det(L_1^{(N')}) = 4(\mu + 2) \det(L^{(N'-2)})$ and hence
\begin{equation} \label{lmu1detodd}
        \det(L_1^{(N')}) = 4(\mu + 2)(\mu^2 - 4)^{\frac{N'-3}{2}}(\mu - 1).
\end{equation}
More generally, for any $1 < j < \frac{N'}{2}$, we have
\begin{equation} \label{lmujdetodd}
        \det L_{N-j}^{(N')} = \det(L_j^{(N')}) = \det(L_1^{(N')}).
\end{equation}
It remains to compute $\det L_{\frac{N'+1}{2}}^{(N')}$. Expanding $\det L_{\frac{N'+1}{2}}^{(N')}$ by the first column and then the last column, we obtain the recursion relation
        \[ \det L_{\frac{N'+1}{2}}^{(N')} = (\mu^2 - 4) \det L_{\frac{(N'-2)+1}{2}}^{(N'-2)}.
\]
Together with $\det L_2^{(3)} = \det \left( \begin{array}{ccc} \mu & 4 & -2 \\ 0 & 4 & 0 \\ -2 & 4 & \mu \end{array} \right) = 4(\mu^2 - 4)$, this implies
\begin{equation} \label{lmuspecialodd}
        \det L_{\frac{N'+1}{2}}^{(N')} = 4 (\mu^2 - 4)^{\frac{N'-1}{2}}.
\end{equation}
Hence, combining (\ref{lmudetodd})-(\ref{lmuspecialodd}) we obtain
\begin{eqnarray*}
  \det(P^D - \l \textrm{Id}) & = & (\mu^2 \!\!-\!\! 4)^{\frac{N'\!\!-\!\!3}{2}} \cdot \left( (\mu^2 \!\!-\!\! 4)(\mu \!\!-\!\! 1) + (N'\!\!-\!\!1) \cdot 4(\mu \!\!+\!\! 2)(\mu \!\!-\!\! 1) + 4 (\mu^2 \!\!-\!\! 4) \right) \\
  & = & (\mu^2 - 4)^{\frac{N'-3}{2}} (\mu + 2) (\mu^2 + (4N'-3) \mu - (4N'+2)).
\end{eqnarray*}
Hence, if $N'$ is odd, the zeroes of $\det(P^D - \l \textrm{Id})$ are $\mu = 2$ (with multiplicity $\frac{N'-3}{2}$), $\mu = -2$ (with multiplicity $\frac{N'-1}{2}$), and
        \[ \mu = -\frac{1}{2}(4N'-3) \pm \frac{1}{2} \sqrt{16 N'^2 - 8N' + 17}
\]
(each with multiplicity $1$). Transforming back to $\l = -1-\mu$, we obtain the claimed formulas for the eigenvalues in the case where $N'$ is odd.
This completes the proof of Lemma \ref{eigenvpdtheorem}.
\end{proof}

\vspace{.4cm}

\textsc{Institut f\"ur Mathematik, Universit\"at Z\"urich, Winterthurerstrasse 190, CH-8057 Z\"urich, Switzerland} \\
\emph{E-mail address:} \texttt{andreas.henrici@math.unizh.ch}

\vspace{.4cm}

\textsc{Institut f\"ur Mathematik, Universit\"at Z\"urich, Winterthurerstrasse 190, CH-8057 Z\"urich, Switzerland} \\
\emph{E-mail address:} \texttt{thomas.kappeler@math.unizh.ch}

\end{document}